\documentclass[10pt,journal,a4paper]{IEEEtran}
\usepackage{amssymb,amsmath}
\usepackage{cite}
\usepackage{graphicx,subfigure}
\usepackage{psfrag}
\usepackage{url}
\usepackage{multirow}
\usepackage[latin1]{inputenc}
\usepackage[absolute,overlay]{textpos}

\usepackage{pgf}
\usepackage[latin1]{inputenc}
\usepackage{bm}
\usepackage{stackengine}
\newcommand\stackequal[2]{%
	\mathrel{\stackunder[1pt]{\stackon[2pt]{$\gtrless$}{$\scriptscriptstyle#1$}}{%
			$\scriptscriptstyle#2$}}}

\usepackage{array}
\newcolumntype{P}[1]{>{\centering\arraybackslash}p{#1}}
\usepackage{booktabs}
\usepackage{balance}

\usepackage{amsthm}

\newtheorem{lemma}{Lemma}
\newtheorem{theorem}{Theorem}

\newtheorem{remark}{Remark}

\begin{document}

\title{Statistical Analysis of Multiple Antenna Strategies for Wireless Energy Transfer}
\author{
	\IEEEauthorblockN{Onel L. A. López, 
						Hirley Alves, 
						Richard D. Souza, 
						Samuel Montejo-Sánchez
					}
	\thanks{Onel L. A. L\'opez and Hirley Alves are with the Centre for Wireless Communications (CWC), University of Oulu, Finland. \{onel.alcarazlopez,hirley.alves\}@oulu.fi}
	\thanks{Richard D. Souza is with Federal University of Santa Catarina (UFSC), Florianópolis, Brazil. \{richard.demo@ufsc.br\}.}
	\thanks{S. Montejo-Sánchez is with Programa Institucional de Fomento a la I+D+i, Universidad Tecnológica Metropolitana, Santiago, Chile. \{smontejo@utem.cl\}.}
	\thanks{This work is supported by Academy of Finland (Aka) (Grants n.307492, n.318927 (6Genesis Flagship), n.319008 (EE-IoT)), as well as Finnish Foundation for Technology Promotion, FONDECYT Postdoctoral Grant n.3170021 and CNPq Grant n.698503P (PrInt) (Brazil).}  
}
					
\maketitle

\begin{abstract}	
	Wireless Energy Transfer (WET) is emerging as a potential solution for powering small energy-efficient devices.
	We propose strategies that use multiple antennas at a power station, which wirelessly charges a large set of single-antenna devices. Proposed strategies operate without Channel State Information (CSI) and we attain the distribution and main statistics of the harvested energy under Rician fading channels with sensitivity and saturation energy harvesting (EH) impairments.
	A switching antenna strategy, where a single antenna with full power transmits at a time, provides the most predictable energy source, and it is particularly suitable for powering sensor nodes with highly sensitive EH hardware operating under non-LOS (NLOS) conditions; while other WET schemes perform alike or better in terms of the average harvested energy. 
	Under NLOS switching antennas is the best, while when LOS increases transmitting simultaneously with equal power in all antennas is the most beneficial. Moreover, spatial correlation is not beneficial unless the power station transmits simultaneously through all antennas, raising a trade-off between average and variance of the harvested energy since both metrics increase with the spatial correlation.
	Moreover, the performance gap between CSI-free and CSI-based strategies decreases quickly as the number of devices increases.
\end{abstract}
\section{Introduction}
With the advent of the Internet of Things (IoT) era, there is an increasing interest in energy efficient technologies in order to prolong the battery life time of the devices. The recent trends in energy harvesting (EH) techniques provide a fundamental efficient method that avoids replacing or recharging batteries, which may be costly, inconvenient or hazardous, e.g., in toxic environments, for sensors embedded in building structures or inside the human body \cite{Zhang.2013}. Many types of EH schemes, according to the energy source, have been considered, based on solar, piezoelectric, wind, hydroelectric, and wireless radio frequency (RF) signals \cite{Sudevalayam.2011}. While harvesting energy from environmental sources is dependent on the presence of the corresponding energy source, RF-EH
provides key benefits in terms of being wireless, readily available in the form of transmitted energy (TV/radio broadcasters, mobile base stations and handheld radios), low cost, and having small form factor implementation. Hence, it becomes potentially suitable for a variety of low-power applications such as wireless sensor networks (WSNs) and RF identification (RFID) networks \cite{Valenta.2014}\footnote{In fact, commercial WET-enabled sensors and RFID tags are already in the market; check for instance \url{www.powercastco.com}.
Readers are also encouraged to refer to \cite{Boyer.2014} and references therein for a complete overview of the specific characteristics and challenges of RFID networks. Therein, authors re-examine notions of  power and spectral efficiency from an energy-constraint perspective and establish the trade-offs between diversity order and spatial multiplexing gains.}. 

Three main transmit scenarios can be distinguished in RF-EH networks\footnote{In this work we consider those RF-EH networks where the RF signals are intentionally transmitted for powering the EH devices. Alternatively, wireless EH refers to those setups where devices harvest energy from non-dedicated RF signals.}, namely Wireless Energy Transfer (WET) \cite{Grover.2010}, Wireless Powered Communication Network (WPCN) \cite{Bi.2015} and Simultaneous Wireless Information and Power Transfer (SWIPT) \cite{Varshney.2008}. In the first scenario a power transmitter transfers energy to EH receivers to charge their batteries, without any information exchange, while WPCN refers to those cases where  the EH receiver uses the energy harvested in a first phase to transmit its information in a second phase. Finally, in the third scenario a hybrid transmitter is transferring wireless energy and information signals using the same waveform to multiple receivers. More details on each of these scenarios can be found in \cite{Alsaba.2018}, along with a survey on energy beamforming (EB) techniques. In this work we focus on WET scenarios, that also could be seen as an element of WPCN systems\footnote{This is because WPCN setups include WET, which is followed by a wireless information transfer phase.}, while readers can refer to \cite{Clerckx.2018} for a review and discussion on recent progress on SWIPT technologies.
\subsection{Related Work}\label{IntA}
Many recent works have considered specifically WET and WPCN setups in different contexts and scenarios. An overview of the key networking structures and performance enhancing techniques to build an efficient WPCN is provided in \cite{Bi.2016}, while authors also point out new and challenging research directions. A power beacon (PB) that constantly broadcasts wireless energy in a cellular network for RF-EH was proposed in \cite{Huang.2014}. These PBs are deployed in conjunction with base stations to provide power coverage and signal coverage in the network, while the deployment of this hybrid network under an outage constraint on data links was designed using stochastic-geometry tools. In \cite{Ju.2014}, a hybrid access point (AP) was proposed where the AP broadcasts wireless power in the downlink followed by data transmission using the harvested energy in the uplink in a time-division
duplex (TDD) manner. Also in TDD setups, works in \cite{Lopez.2017,Lopez.2018,Lopez2.2018} consider the transmission of separately short energy and information packets (stringent delay constraints) in ultra-reliable WPCN scenarios under different channel conditions, e.g., Rayleigh or Nakagami-m fading. Authors either analyze the performance of the information transmission phase \cite{Lopez.2017}, or optimize it by using power \cite{Lopez.2018} control or cooperative schemes \cite{Lopez2.2018}. Some scheduling strategies that allow a direct optimization of the energy efficiency of the network are also proposed in \cite{Bacinoglu.2018}. Additionally, an energy cooperation scheme that enables energy cooperation in battery-free wireless networks with WET is presented in \cite{Li.2018}. 

Yet, WET requires shifts in the system architecture and in its resource allocation strategies for optimizing the energy supplying, thus, avoiding energy outages. In that regard, authors in \cite{Chen.2017} study the  probability density function (PDF), the cumulative distribution function (CDF), and the average of the energy harvested from signals transmitted by multiple sources. Interestingly, such information allows to determine the best strategies when operating under different channel conditions. Additionally,  multi-antenna EB, where the energy-bearing signals are weighted at the multiple transmit antennas before transmission, has been proposed very recently \cite{Huang.2016,Park.2017}. The average throughput performance of EB in a WPCN, consisting of one hybrid AP with multiple antennas and a single-antenna user, is investigated in \cite{Huang.2016}. The impact of various parameters, such as the AP transmit power, the energy harvesting time, and the number of antennas in the system throughput is analyzed. In \cite{Park.2017}, authors propose an EB scheme that maximizes the weighted sum of the harvested energy and the information rate in multiple-input single-output (MISO) WPCN. They show that their proposed scheme achieves the highest performance compared to existing work. In practice, the benefits of EB in WET crucially depend on the available CSI at the transmitter. 
An efficient channel acquisition method for a point-to-point multiple-input multiple-output (MIMO) WET system is designed in \cite{Zeng.2015} by exploiting channel reciprocity. Authors provide useful insights on when channel training should be employed to improve the transferred energy. Meanwhile, the training design problem is studied in \cite{Zeng2.2015} for MISO WET systems in frequency-selective channels.
\subsection{Contributions and Organization of the Paper}\label{cont}
The problem of CSI acquisition in WET systems is critical and limits the practical significance of some previous work. This is because WET systems are inherently energy-limited, and part of the harvested energy would need to be used for CSI acquiring purpose \cite{Hu.2018}. In fact, the required energy resources to that end cannot be neglected when there is a large number of antennas and/or if the estimation takes place at the EH side since it requires complex baseband signal processing. Even when previous problems could be addressed in some particular scenarios, there is still the problem of CSI acquisition in multi-user setups, specially in IoT use cases where the broadcast nature of wireless transmissions could be exploited for powering a massive number of devices simultaneously. In such cases, effective CSI-free strategies are of vital importance. 

This paper addresses CSI-free WET with multiple transmit antennas, while assuming practical characteristics of EH hardware. The main contributions of this work can be listed as follows:
\begin{itemize}
	\item We present and analyze multiple antenna strategies at a dedicated PB serving a large set of RF-EH devices without any CSI. We do not consider any other information related to devices such as topological deployment, battery charge; although, such information could be crucial in some setups\footnote{For instance, the PB could prioritize those devices that are more distant and/or with low battery charge by focusing more energy on their directions.}. Our derivations are specifically relevant for scenarios where it is difficult and/or not worth obtaining such information, e.g., when powering a massive number of low-power EH devices uniformly distributed in a given area and possibly with null/limited feedback to the PB. Additionally, the performance analysis considers the harvested energy at the receiver, and comparisons with ideal CSI-based schemes are carried out; 
	\item We attain the distribution and some main statistics of the harvested energy in correlated Rician fading channels under each WET scheme and ideal EH operation.  These results are extended to more practical scenarios where EH sensitivity and saturation impairments come to play. The Rician fading assumption is general enough to include a class of channels, ranging from Rayleigh fading channel without line of sight (LOS) to a fully deterministic LOS channel, by varying the Rician factor $\kappa$;	
	\item When powering a given sensor, we found that switching antennas such that only one antenna with full power transmits at a time, guarantees the lowest variance in the harvested energy, thus providing the most predictable energy source. This technique is particularly suitable under highly sensitive EH hardware and when operating under non LOS (NLOS) conditions;  
	\item While under NLOS it is better switching antennas, under some LOS it is better transmitting simultaneously with equal power by all antennas. Also, latter strategy provides the greatest fairness when powering multiple sensors when there is a large number of antennas since  devices are more likely to operate near or in saturation. Additionally, an increase in the spatial correlation is generally prejudicial, except when transmitting simultaneously with equal power by all antennas, for which there is a trade-off between average and variance of the harvested energy since both metrics increase with the spatial correlation; 	
	\item Numerical results validate our analytical findings and demonstrate the suitability of the CSI-free over the CSI-based strategies as the number of devices increases.
\end{itemize}

Next, Section~\ref{system} presents the system model, while Section~\ref{wet} introduces the WET strategies under study. Their performance under Rician fading is investigated in Section~\ref{Ric}, while Section~\ref{results} presents numerical results. Finally, Section \ref{conclusions} concludes the paper.
\begin{table*}[!t]
	\caption{Main Symbols}
	\label{notat}
	\centering
	\begin{tabular}{p{1cm} p{6.8cm}|p{1cm} p{6.8cm}}
		\toprule
		$T,\mathcal{S}$ & Dedicated power station and set of sensor nodes&
		$M$ & Number of antennas at $T$\\
		$S_j$ & $j-$th sensor node ($j-$th element of $\mathcal{S}$) & $\mathbf{h}$ & Complex fading channel vector\\
		$\bm{\alpha},\bm{\beta}$ & Real and imaginary parts of $\mathbf{h}$, respectively 
		& $l$  & Number of energy beams transmitted by $T$ \\
		$\mathbf{w}_k$ & $k-$th precoding vector for the $k-$th energy beam & $\varpi_1,\varpi_2$ & Sensitivity and saturation levels of the EH hardware  \\
		$\varrho_j$ &  Path loss of the link $T\rightarrow S_j$ times the overall transmit power of $T$& $g$ & Function that describes the relation between harvested energy and RF input power \\
		$\xi_j^{\mathrm{rf}},\xi_j$ & Incident RF energy and harvested energy at $S_j$ &
		$\eta$ & Energy conversion efficiency \\
		$\xi^0$ & Harvested energy under the ideal EH linear model & $\kappa$ & Rician fading factor \\ 
		$\bm{\mu}$ & Mean vector of real and imaginary parts of $\mathbf{h}$ &
		$\mathbf{R}$ & Normalized covariance matrix of elements of $\bm{\alpha}$, $\bm{\beta}$ \\
		$\rho$ & Uniform spatial correlation coefficient & $\tau$ & Exponential spatial correlation coefficient\\
		$\delta$& General correlation parameter & $\sigma^2$ & Variance of the elements of $\bm{\alpha}$ and $\bm{\beta}$\\		
		$\lambda$ & Eigenvalue of $\mathbf{R}$ &
		$\bm{\Lambda}$ & Diagonal matrix containing the values of $\lambda$ \\
		$\mathbf{B}$ & Matrix with orthogonalized eigenvectors of $\mathbf{R}$ &	
		$\varphi,\psi$ & Non central chi-squared distribution parameters\\
		\begin{tabular}[l]{@{}l@{}} $\xi_\mathrm{th}$ \\ \  \\ $p_1,p_2,p_3$ \end{tabular}  & \begin{tabular}[l]{@{}l@{}} Minimum amount of RF input power for which $S$\\ can operate\\ Curve fitting parameters of benchmark EH model \end{tabular} &		
		\begin{tabular}[l]{@{}l@{}} $\kappa_1^*,\kappa_2^*$\\ \ \\ \ \end{tabular} & \begin{tabular}[l]{@{}l@{}} $\kappa$ that maximizes the variance of the harvested\\ energy under the $\mathrm{AA}$ scheme and $\mathrm{SA}$ or $\mathrm{AA\!-\!CSI}$\\ schemes, respectively\end{tabular} \\	
		\bottomrule		
	\end{tabular}
\end{table*}

\textbf{Notation:} Boldface lowercase letters denote vectors, while boldface uppercase letters denote matrices. For instance, $\mathbf{x}=\{x_i\}$, where $x_i$ is the $i$-th element of vector $\mathbf{x}$; while $\mathbf{X}=\{x_{i,j}\}$, where $x_{i,j}$ is the $i$-th row $j$-th column element of matrix $\mathbf{X}$. By $\mathbf{I}$ we denote the identity matrix, and by $\mathbf{1}$ we denote a vector of ones. The superscript $(\cdot)^T$ denotes the transpose, $\mathrm{det}(\cdot)$ the determinant, and by $\mathrm{Diag}[x_1,x_2,\cdots]$ we denote the diagonal matrix with elements $x_1,x_2,\cdots$. The $\ell_p-$norm of vector $\mathbf{x}$ is  $||\ \!\mathbf{x}\!\ ||_p=\big(\sum_{i}|x_i|^p\big)^{1/p}$ \cite[Eq.(3.2.13)]{Thompson.2011}. $\mathbb{C}$ is the set of complex numbers and $\mathbf{i}=\sqrt{-1}$ is the imaginary unit. Meanwhile, $x^*$ is the conjugate of $x$, and $|\ \!\!\cdot\!\!\ |$ is the absolute operation, or cardinality of the set, according to the case. $\mathbb{E}[\!\ \cdot\ \!]$ and $\mathrm{VAR}[\!\ \cdot\ \!]$ denote expectation and variance, respectively, while $\mathbb{P}[A]$ is the probability of event $A$. $\mathbf{y}\sim\mathcal{N}(\bm{\mu},\mathbf{R})$ is a Gaussian random vector with $\mathbb{E}[\bm{y}]=\bm{\mu}$ and covariance $\mathbf{R}$,  $W\sim\mathrm{RIC}(\kappa)$ is a Rician random variable (RV) with factor $\kappa$ \cite[Ch.2]{Proakis.2001}, while $V\sim\Gamma(m,a/m)$ is a gamma random variable with PDF and CDF given by
\begin{align}
f_V(v)&=\frac{(m/a)^m}{\Gamma(m)}v^{m-1}e^{-mv/a},\ v\ge 0,\label{fV}\\ F_V(v)&=1-\frac{\Gamma(m,mv/a)}{\Gamma(m)},\ v\ge 0, \label{FV}
\end{align}
where $\Gamma(p)$ and $\Gamma(p,x)$ are the complete and incomplete gamma functions, respectively. Additionally, $Z\sim\chi^2(\varphi,\psi)$ is a non-central chi-squared RV with $\varphi$ degrees of freedom and parameter $\psi$, thus, its PDF and CDF are given by \cite[Eqs.(2-1-118) and (2-1-121)]{Proakis.2001}
\begin{align}
f_Z(z)&=\!\frac{1}{2}e^{-(z\!+\!\psi)/2}\Big(\frac{z}{\psi}\Big)^{\varphi/4\!-\!1/2}I_{\varphi/2\!-\!1}\big(\sqrt{\psi z}\big),\ z\ge 0,\label{fZ}\\ F_Z(z)&=1-\mathcal{Q}_{\varphi/2}\big(\sqrt{\psi},\sqrt{z}\big),\ z\ge 0,\label{FZ}
\end{align}
where $I_{n}$ is the $n$-th order modified Bessel function of the first kind \cite[Eq.(2-1-120)]{Proakis.2001} and $\mathcal{Q}$ is the Marcum Q-function \cite[Eq.(2-1-122)]{Proakis.2001}. According to \cite[Eq.(2-1-125)]{Proakis.2001} 
\begin{align}
\mathbb{E}[Z]=\varphi+\psi,\qquad \mathrm{VAR}[Z]=2(\varphi+2\psi).\label{meanvariance}
\end{align}
Table~\ref{notat} summarizes the main symbols used throughout this paper.
\section{System Model}\label{system}
Consider the scenario in Fig.~\ref{Fig1}, in which a dedicated power station $T$ equipped with $M$ antennas, powers wirelessly a large set $\mathcal{S}=\{S_1,S_2,\cdots,S_{|\mathcal{S}|}\}$ of single-antenna sensor nodes located nearby. Quasi-static channels are assumed, where the fading process is considered to be constant over the transmission of a block  and independent and identically distributed (i.i.d) from block to block. The fading channel coefficient between the $i$-th antenna of $T$ and the $j-$th sensor node $S_j$ is denoted as $h_{i,j}\in\mathbb{C}$, while $\mathbf{h}_{j}\in\mathbb{C}^{M\times 1}$ is a vector with the channel coefficients from the power station antennas to $S_j$.

In general, during WET power station $T$ may transmit with up to $l\le M$ energy beams to broadcast energy to all sensors in $\mathcal{S}$. Then, the incident RF power at the $j-$th EH receiver is given by\footnote{Notice that the EH receiver does not need to convert the received signal from the RF band to the baseband in order to harvest the carried energy. Nevertheless, thanks to the law of energy conservation, it can be assumed that the total incident RF power (energy normalized by the baseband symbol period) matches the power of the equivalent received baseband signal.}
\begin{align}\label{E}
\xi^{\mathrm{rf}}_j=\varrho_j\Big|\mathbf{h}_j^T\sum_{k=1}^{l}\mathbf{w}_kx_k\Big|^2=\varrho_j\sum_{k=1}^{l}\big|\mathbf{h}_j^T\mathbf{w}_k\big|^2,
\end{align}
where $\mathbf{w_k}\in\mathbb{C}^{M\times 1}$ denotes the precoding vector for generating the $k-$th energy beam, and $x_k$ is its energy-carrying signal. Without loss of generality, we assume that each $x_k$ is taken independently from an arbitrary distribution with $\mathbb{E}[|x_k|^2]=1, \forall k$ \cite{Ju.2014,Chen.2017,Zeng.2015}. Additionally, we set $\sum_{k=1}^{l}||\mathbf{w}_k||_2^2=1$, while $\varrho_j$ accounts for the path loss of the link $T\rightarrow S_j$ times the overall transmit power of $T$. 
By considering negligible noise energy, the harvested energy\footnote{We use the terms energy and power indistinctly, which can be interpreted as if block duration is normalized.}, $\xi_j$, can be written as a function of $\xi^{\mathrm{rf}}_j$ as $\xi_j=g(\xi^{\mathrm{rf}}_j)$ where $g:\mathbb{R} \rightarrow \mathbb{R}$ is a non-decreasing function of its argument. In general $g$ is nonlinear\footnote{Although in this work we just state the harvested energy as a function of the incident RF power, we would like to highlight that in practice $g$ also depends on the modulation and incoming waveform \cite{Zeng.2017}, and some approaches have been considered in the literature for exploiting this. For instance \cite{Clerckx.2018_2}: energy waveform, which relies on deterministic multisine at the transmitter so as to excite the rectifier in a periodic manner by sending high energy pulses; and energy modulation, which relies on single-carrier transmission with inputs following a probability distribution that boosts the fourth order term in the Taylor expansion of the diode characteristic function.} 
 and mathematical analyses are cumbersome, but starting from the linear model the accuracy can be significantly improved by considering  three main factors that limit strongly the performance of a WET receiver \cite{Valenta.2014,Clerckx2.2018,Zhou.2013,Lopez.2018}: \textit{(i)} its sensitivity  $\varpi_1$, which is the minimum RF input power required for energy harvesting; \textit{(ii)} its saturation level $\varpi_2$, which is the RF input power for which the diode starts working in the breakdown region, and from that point onwards the output DC power keeps practically constant; and \textit{(iii)} the energy efficiency $\eta\in[0,1]$ in the interval $\varpi_1 \le \xi^{\mathrm{rf}}_j\le \varpi_2$, which we assume as constant. Therefore, we can write $\xi_j$ as
\begin{align}\label{conv}
\xi_j=g(\xi_{j}^{\mathrm{rf}})=\left\{ \begin{array}{lrr}
0, & & \xi^{\mathrm{rf}}_j<\varpi_1 \\
\eta \xi^{\mathrm{rf}}_j, &  &\varpi_1\le\xi^{\mathrm{rf}}_j<\varpi_2 \\
\eta\varpi_2,&    &\xi^{\mathrm{rf}}_j\ge \varpi_2
\end{array}
\right..
\end{align}
Notice that the linear model assumed in most of the related literature, e.g., \cite{Zhang.2013,Alsaba.2018,Huang.2014,Ju.2014,Lopez.2017,Lopez2.2018,Bacinoglu.2018,Huang.2016,Park.2017,Zeng.2015,Zeng2.2015,Hu.2018,Zhou.2013}, does not take into account the sensitivity and saturation phenomena, which is equivalent to operate with $\varpi_1\!=\!0$ and $\varpi_2\!\rightarrow\!\infty$. Thus, taking these impairments into consideration in the scenarios under discussion is an important contribution from a practical perspective.
\begin{figure}[t!]
	\centering  \includegraphics[width=0.85\columnwidth]{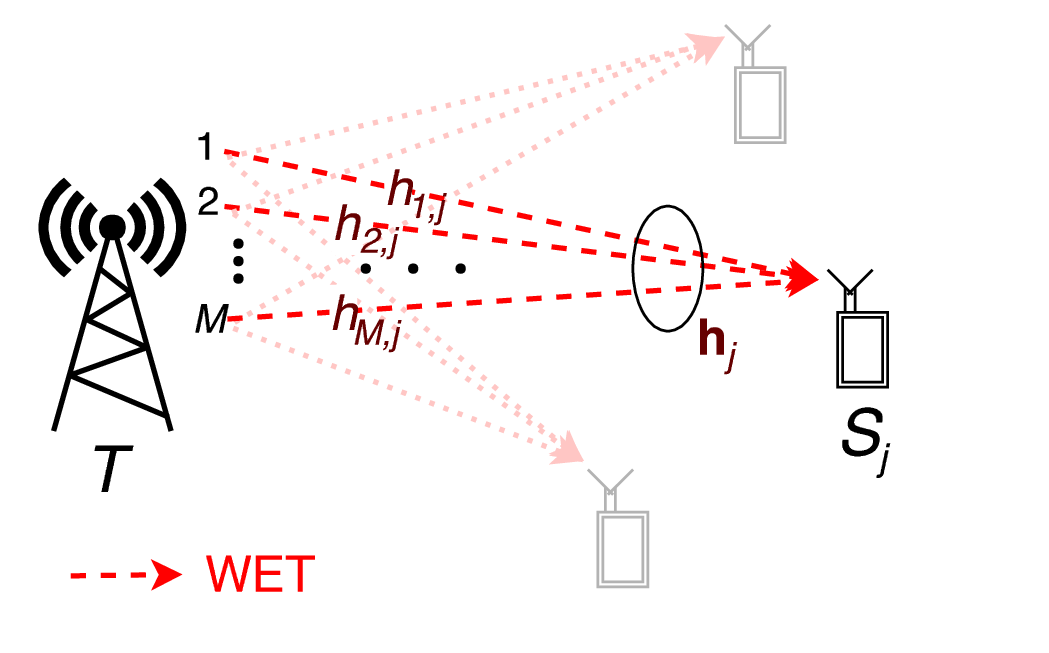}
	\caption{System model: dedicated power station $T$ equipped with $M$ antennas, powers wirelessly a set $\mathcal{S}$ of single-antenna sensor nodes located nearby.}		
	\label{Fig1}
\end{figure}
\section{WET Strategies}\label{wet}
First, in Subsection~\ref{nn} we characterize the performance of WET for three different strategies at $T$ without any CSI, while two alternative strategies that require full CSI are presented as benchmarks in Subsection~\ref{ss}.
\subsection{WET Strategies without CSI}\label{nn}
Since no CSI is available and $\mathbf{w}_k$ cannot depend on the channel coefficients, $T$ does not form energy beams  to reach efficiently each $S_j$. Therefore, for these kind of strategies it is only necessary focusing on the performance of an arbitrary user $S_j$, while also setting $l=1$. 
\subsubsection{One Antenna ($\mathrm{OA}$)}
Is the simplest strategy because only one out of $M$ antennas transmitting with full power is used for powering the devices. 
Then, using \eqref{E} we obtain
\begin{align}\label{1a}
	\xi_{j_\mathrm{OA}}=g\big(\varrho_j|h_{i,j}|^2\big),	
\end{align}
where $i\in\{1,2,\cdots,M\}$. Notice  that in this case $\mathbf{w}_1$ is a vector of zeros with entry $1$ in the $i-$th element. There is no difference whether $T$ is equipped with only one or several antennas when operating with the $\mathrm{OA}$ strategy. 
\subsubsection{All Antennas at Once ($\mathrm{AA}$)}
The $\mathrm{OA}$ strategy does not exploit multiple antennas, thus, it does not take advantage of that degree of freedom. One obvious and simple alternative is transmitting with all antennas but with reduced power at each, $\mathbf{w}_1 = (1/\sqrt{M})\mathbf{1}_{M\times 1}$, thus,\footnote{Notice that any other power allocation is not advisable since $T$ does not know how the channels are behaving.}
\begin{align}\label{aao}
\xi_{j_\mathrm{AA}}=g\bigg(\frac{\varrho_j}{M}\Big|\sum_{i=1}^{M}h_{i,j}\Big|^2\bigg).
\end{align}
The $\mathrm{OA}$ and $\mathrm{AA}$ schemes are the extreme cases of a more general strategy where $K$ out of $M$ antennas are selected to power the sensors. As a consequence, the $K$-out-$M$'s performance is limited by that of the $\mathrm{OA}$ and $\mathrm{AA}$ strategies.
\subsubsection{Switching Antennas ($\mathrm{SA}$)}
Instead of transmitting with all antennas at once, $T$ may transmit with full power by one antenna at a time such that all antennas are used during  a block. Assuming equal-time allocation for each antenna, the system is equivalent to that in which each sub-block duration is $1/M$ of the total block duration, and the total harvested energy accounts for the sum of that of the $M$ sub-blocks. That is
\begin{align}\label{oat}
\xi_{j_\mathrm{SA}}=\frac{1}{M}\sum_{i=1}^{M}g\big(\varrho_j|h_{i,j}|^2\big).
\end{align}
Note that $\mathbf{w}_1$ in each sub-block is given as in the $\mathrm{OA}$ strategy, but the chosen $i$ is different in each sub-block.
\subsection{Benchmark WET Strategies}\label{ss}
For the sake of describing some benchmark strategies, herein we consider that full CSI is available at $T$. Sensors $\mathcal{S}$ could use a ``small'' amount of energy\footnote{Coming from a short phase where $\mathcal{S}$ are first powered by one of the previous CSI-free strategies, or even from some residual energy after previous rounds.} to send some pilot symbols to $T$ in order to acquire the CSI. This requires reciprocal channels, thus, $T$ should be listening at the time of the transmission. Otherwise, $T$ has to send the pilots and wait for a feedback from the sensor(s) informing the CSI. Whatever the case, it seems unsuitable in a setup where a large number of energy constrained sensors require the powering service from $T$, specially if we also consider the multiple access problem. Therefore, we present the following CSI-based strategies only as benchmarks for those presented in the previous subsection. We assume that $T$ knows also the EH hardware characteristics, e.g., $\varpi_1,\varpi_2,\eta,$ of the EH sensors and the goal is to maximize the overall energy harvested by $\mathcal{S}$. 
\subsubsection{Best Antenna ($\mathrm{OA-CSI}$)}
This strategy is the counterpart for the previous $\mathrm{OA}$ scheme when full CSI is available. In this case, the antenna that provides the greatest amount of overall harvested energy is selected out of the overall set. Therefore, $l=1$ and
\begin{align}\label{1acsi}
\sum_{j=1}^{|\mathcal{S}|}\xi_{j_\mathrm{OA-CSI}}=\max_{i=1,...,M}\sum_{j=1}^{|\mathcal{S}|}g\big(\varrho_j|h_{i,j}|^2\big),
\end{align}
since this time $\mathbf{w}_1$ is a vector of zeros with entry $1$ in the  selected antenna index. Of course, in a multi-user system where the nodes benefit from the WET phase simultaneously, the best antenna is usually not the same for all users. Different from \eqref{1acsi}, another possible implementation may be that in which some users are optimized first and then the others, but in all these scenarios the complexity scales quickly with the number of users while reducing the overall performance.
\subsubsection{Best Transmit Beamforming ($\mathrm{AA-CSI}$)}\label{best}
This strategy is the counterpart for the previous $\mathrm{AA}$ scheme when full CSI is available. In this case, instead of transmitting with the same power over each antenna, $T$ precompensates through $\mathbf{w}_k$ for the channel and EH hardware effects before transmission such that the overall harvested energy at $\mathcal{S}$ is maximized. Hence, 
\begin{align}\label{aaocsi}
\sum_{j=1}^{|\mathcal{S}|}\xi_{j_\mathrm{AA-CSI}}=\max_{\{\mathbf{w}_k\}}\sum_{j=1}^{|\mathcal{S}|}g\Big(\varrho_j\sum_{k=1}^{M}\big|\mathbf{h}_j^T\mathbf{w}_k\big|^2\Big),
\end{align}
where $l$ is set to $M$ since that is the maximum possible number of energy beams. In case the system performance maximizes with a smaller $l$, which may be the case when $|\mathcal{S}|<M$, then, some of the optimum beamformers are all-zeros vectors. Notice that the $\mathrm{OA-CSI}$ scheme is less sensitive to CSI imperfections than the $\mathrm{AA-CSI}$. This is because the former only relies on the power gain of the channel, while the latter requires the full characterization, envelope and phase, of the channel coefficients.
\subsection{Comparison of the WET Strategies for $|\mathcal{S}|=1$ and under Ideal EH Linear Model}
Some useful insights come from setting $|\mathcal{S}|\!=\!1$ and using the ideal EH linear model such that $\varpi_1\!=\!0$ and $\varpi_2\!\rightarrow\!\infty$. In this case, we denote $\xi^0$ as the harvested energy at the unique sensor node while we avoid using the subindex $j$, then, \eqref{1a}, \eqref{aao}, \eqref{oat}, \eqref{1acsi} and \eqref{aaocsi} can be rewritten as
\begin{align}
\xi^0_{\mathrm{OA}}&=\eta \varrho |h_{i}|^2 		\label{1av2},\\
\xi^0_{\mathrm{AA}}&=\frac{\eta\varrho}		{M}\Big|\sum_{i=1}^{M}h_{i}\Big|^2\label{aaov2},\\
\xi^0_{\mathrm{SA}}&=\frac{\eta\varrho}{M}\sum_{i=1}^{M}|h_{i}|^2,\label{oatv2}\\ \xi^0_{\mathrm{OA-CSI}}&=\eta\varrho\max_{i=1,\cdots,M}|h_{i}|^2,\label{1acsiv2}\\		\xi^0_{\mathrm{AA-CSI}}&=\eta\varrho\max_{\mathbf{w}_1}|\mathbf{h}^T\mathbf{w}_1|^2\stackrel{(a)}{=}\eta\varrho\sum_{i=1}^{M}|h_{i}|^2,\label{aaocsiv2}
\end{align}
where $(a)$ comes from setting $\mathbf{w}_1=\mathbf{h}^*/||\mathbf{h}||_2$ which is the optimum precoding vector for Maximum Ratio Transmission  (MRT) \cite{Lo.1999} in a MISO system.
\begin{theorem}\label{the1}	
	The following relations are satisfied:
	\begin{align}
	\xi^0_{\mathrm{OA}}&\le \xi^0_{\mathrm{OA-CSI}}\le\xi^0_{\mathrm{AA-CSI}}= M\xi^0_{\mathrm{SA}},\label{eq1}\\
	\xi^0_{\mathrm{AA}}&\le\xi^0_{\mathrm{AA-CSI}}= M\xi^0_{\mathrm{SA}}.\label{eq2}
	\end{align}	
\end{theorem}
\begin{proof}
	According to \eqref{1av2}, \eqref{1acsiv2}, \eqref{aaocsiv2} and \eqref{oatv2} we have that
	\begin{align}
	\xi^0_{\mathrm{OA}}&=\eta\varrho|h_{i}|^2\le \eta\varrho\max_{i=1,...,M}|h_{i}|^2=\xi^0_{\mathrm{OA-CSI}}\nonumber\\
	&\stackrel{(a)}{\le} \eta\varrho\sum_{i=1}^{M}\big|h_{i}\big|^2=\xi^0_{\mathrm{AA-CSI}}=M\xi^0_{\mathrm{SA}},
	\end{align}
	where the equality in $(a)$ is only when $M=1$. Thus, \eqref{eq1} is satisfied. Now, for the second part of the proof we proceed from \eqref{aaov2}, \eqref{aaocsiv2} and \eqref{oatv2} as follows.
	\begin{align}
	\frac{\xi^0_{\mathrm{AA}}}{\eta\varrho}&=\frac{1}{M}\bigg|\sum_{i=1}^{M}h_{i}\bigg|^2\stackrel{(a)}{\le}\frac{1}{M}\bigg(\sum_{i=1}^{M}\big|h_{i}\big|\bigg)^2\nonumber\\
	&\stackrel{(b)}{=}\frac{\big|\big|\mathbf{h}\big|\big|_1^2}{M}\stackrel{(c)}{\le}\big|\big|\mathbf{h}\big|\big|_2^2\stackrel{(b)}{=}\sum_{j=1}^{M}\big|h_{i}\big|^2\nonumber\\
	&=\frac{\xi^0_{\mathrm{AA-CSI}}}{\eta\varrho}=\frac{M\xi^0_{\mathrm{SA}}}{\eta\varrho},\label{ME}
	\end{align}
	where $(a)$ comes from applying the triangular inequality and generalizing as shown in \cite[Section~1.1.7]{Andreesc.2006}, 
	$(b)$ follows from using each time the $\ell_p-$norm notation, while $(c)$ from the inequality between the arithmetic and quadratic mean. 
\end{proof}
\section{Analysis under Rician Fading}\label{Ric}
Herein we assume that channels undergo Rician fading, which is a very general assumption that allows  modeling a wide variety of channels by tuning the Rician factor $\kappa\ge 0$, e.g., when $\kappa=0$ the channel envelope is Rayleigh distributed, while when $\kappa\rightarrow\infty$ there is a fully deterministic LOS channel. Since this work deals mainly with CSI-free WET schemes, and for such scenarios the characterization of one sensor is representative of the overall performance, we focus our attention to the case of a generic node, thus avoiding subindex $j$. Additionally, the performance gap between the CSI-free and CSI-based WET schemes is maximum for $|\mathcal{S}|=1$ and analyzing such scenario when CSI is available allows getting analytical expressions along with some useful insights. Previous assumptions  imply that the envelope distribution of $h_{i}$ is Rician distributed with factor $\kappa$, e.g., $|h_{i}|\sim\mathrm{RIC}(\kappa)$ while the channels are Gaussian with independent real and imaginary parts, $\mathbf{h}=\bm{\alpha}+\mathbf{i}\bm{\beta}$ with $\bm{\alpha},\bm{\beta}\sim \mathcal{N}\big(\frac{1}{\sqrt{2}}\bm{\mu},\sigma^2\mathbf{R}\big)$, where $\sigma^2\mathbf{R}$ and $\bm{\mu}$ are respectively the covariance matrix\footnote{Hence, we consider the spatial correlation phenomenon, which occurs due to insufficient spacing between antenna elements, small angle spread, existence of few dominant scatterers, and the antenna geometry. The general concept of spatial correlation is usually linked only to a positive spatial correlation, even when the negative correlation is also physically possible mainly due to the use of decoupling networks and antenna geometry effects. On the other hand, notice that correlation between the coefficients $h_{i}$ is highly probable in the kind of systems we are investigating here because of the short range transmissions \cite{Shan.2000}. Even though, by setting $\mathbf{R}=\mathbf{I}$ we are also able of modeling completely independent fading realizations over all the antennas.} and mean vector of $\bm{\alpha},\bm{\beta}$ \cite{Jayaweera.2005}. Specifically, $\sigma^2$ is the variance of each $\alpha_i,\beta_i$, hence, $|r_{i,j}|\le 1,\ r_{i,i}=1,\ \forall i,j$.

We assume $\bm{\mu}=\mu[\mathbf{1}]_{M\times 1}$, e.g., equal mean over all the fading paths, while factor $\kappa$ is connected to $\mu$ and $\sigma^2$ as 
\begin{align}
\label{kap}
\kappa&=\frac{\mu^2}{2\sigma^2},
\end{align}
and normalizing the channel power gain as $\mu^2+2\sigma^2=1$ \cite{Jayaweera.2005}, e.g., $\mathbb{E}\big[|h_i|^2\big]=1$, we attain
\begin{align}
\sigma^2&=\frac{1}{2(1+\kappa)},\qquad \mu^2=\frac{\kappa}{1+\kappa}.\label{sigm2mu}
\end{align}
Even when in practice $\bm{\mu}=\mu[\mathbf{1}]_{M\times 1}$ does not strictly hold since the channel coefficients $h_i$ usually have different phases on average, such assumption does not impact on the performance of the $\mathrm{OA},\ \mathrm{SA},\ \mathrm{OA-CSI}$ and $\mathrm{AA-CSI}$ schemes. This is because under the $\mathrm{OA},\ \mathrm{SA}$ and $\mathrm{OA-CSI}$ schemes $T$ does not transmit through more than one antenna simultaneously while the instantaneously phase of the signal at the receiver does not affect the EH process. Meanwhile, under the $\mathrm{AA-CSI}$ scheme the PB precompensates through a beamforming vector for both the channel gains and phases such that the average transmission energy is maximized and it was shown according to \eqref{oatv2} and \eqref{aaocsiv2} that $\xi_{\mathrm{AA-CSI}}^0=M\xi_{\mathrm{SA}}$. 
On the other hand, different mean phases do impact on the system performance under the $\mathrm{AA}$ scheme since the energy transmission takes place over all antennas simultaneously and the phases of the signals coming from different antennas could have either a constructive or destructive aggregation at the receiver side which can be more/less accentuated according to their mean values. It is possible to show analytically that the performance of the $\mathrm{AA}$ scheme is maximized when the mean phases of channel coefficients are equal but a deeper discussion is outside of our scope in this article.
In any case, our performance results for the $\mathrm{AA}$ scheme can be considered as upper-bounds for the performance in practical setups, which can be accurate when the mean phases of the channel coefficients are similar.
\subsection{Distribution of the Harvested Energy under Ideal EH Linear Model}\label{linear}
Now, we proceed to characterize the distribution of the harvested energy when using each of the WET strategies analyzed in the previous section under the ideal EH linear model. 
\subsubsection{$\mathrm{OA}$}
Obviously, the spatial correlation has no impact on the performance of the $\mathrm{OA}$ scheme since under that scheme $T$ selects only one antenna without using any information related with the other antennas. From \eqref{1av2} we proceed as follows
\begin{align}
\xi^0_{\mathrm{OA}}&=\eta\varrho|h_{i}|^2=\eta\varrho(\alpha^2+\beta^2)\nonumber\\
&\stackrel{(a)}{=}\eta\varrho \sigma^2\big(\hat{\alpha}^2+\hat{\beta}^2\big)\nonumber\\
&\stackrel{(b)}{\sim}\eta\varrho\sigma^2 \chi^2\big(2,\frac{\mu^2}{\sigma^2}\big)\nonumber\\
&\stackrel{(c)}{\sim} \frac{\eta\varrho}{2(1+\kappa)}\chi^2\big(2,2\kappa\big),
\label{1ar}
\end{align}
where $(a)$ comes from normalizing the variance of $\alpha$ and $\beta$ such that $\hat{\alpha},\hat{\beta}\sim \mathcal{N}\big(\frac{1}{\sqrt{2}}\frac{\mu}{\sigma},1\big)$, $(b)$ comes from the direct definition of a non-central chi-squared RV \cite[Ch.2]{Proakis.2001}, and $(c)$ follows after using \eqref{kap} and \eqref{sigm2mu}. 
\subsubsection{$\mathrm{AA}$}
Now we focus on the performance of the $\mathrm{AA}$ scheme. Based on \eqref{aaov2} we have that
\begin{align}\label{aaor}
\xi^0_{\mathrm{AA}}&=\frac{\eta\varrho}{M}\Big|\sum_{i=1}^{M}\alpha_i+\mathbf{i}\sum_{i=1}^{M}\beta_i\Big|^2\stackrel{(a)}{=}\frac{\eta\varrho}{M}\big|\hat{\alpha}+\mathbf{i}\hat{\beta}\big|^2\nonumber\\
&\stackrel{(b)}{=}\frac{\eta\varrho\sigma^2\delta}{M}\big|\theta+\mathbf{i}\vartheta\big|^2=\frac{\eta\varrho\sigma^2\delta}{M}\big(\theta^2+\vartheta^2\big)\nonumber\\
&\stackrel{(c)}{\sim} \frac{\eta\varrho\sigma^2\delta}{M}\chi^2\Big(2,\frac{M^2\mu^2}{\delta \sigma^2}\Big)\nonumber\\
&\stackrel{(d)}{\sim} \frac{\eta\varrho\delta}{2M(1+\kappa)}\chi^2\Big(2,\frac{2\kappa M^2}{\delta}\Big),
\end{align}
where $(a)$ comes from using $\hat{\alpha}=\sum_{i=1}^{M}\alpha_i$ and $\hat{\beta}=\sum_{i=1}^{M}\beta_i$ and notice that $\hat{\alpha},\hat{\beta}$ are still Gaussian RVs \cite{Novosyolov.2006} with mean $M\mathbb{E}[\alpha_i]=\frac{1}{\sqrt{2}}M\mu$ and variance $\sigma^2\delta$ with
\begin{align}\label{delta}
\delta=\sum_{i=1}^{M}\sum_{j=1}^{M}r_{i,j};
\end{align}
while $(b)$ follows after variance normalization such that $\theta,\vartheta\sim\mathcal{N}\big(\frac{M}{\sqrt{2\delta}}\frac{\mu}{\sigma},1\big)$. Finally, $(c)$  comes from using the definition of a non-central chi-squared RV \cite[Ch.2]{Proakis.2001}, while $(d)$ from using \eqref{kap} and \eqref{sigm2mu}. Notice that  for the special case of 
uniformly spatial correlated fading such that the antenna elements are correlated between each other with coefficient $\rho$ ($r_{i,j}=\rho,\ \forall i\ne j$), we have that
\begin{align}\label{delta2}
\delta=M\big(1+(M-1)\rho\big).
\end{align}
Additionally, in order to guarantee that $\mathbf{R}$ is positive definite and consequently a viable covariance matrix, $\rho$ is lower bounded by $-\frac{1}{M-1}$ \cite{Chen.2015}, thus $-\frac{1}{M-1}\le \rho\le 1$ and $0\le\delta\le M^2$.
\subsubsection{$\mathrm{SA}$}\label{subSA}
From \eqref{oatv2}, we write $\xi^0_{\mathrm{SA}}$ as
\begin{align}
	\xi^0_{\mathrm{SA}}&=\frac{\eta\varrho}{M}\sum_{i=1}^{M}\big(\alpha_i^2+\beta_i^2\big)\nonumber\\
	&=\frac{\eta\varrho}{M}\big(\bm{\alpha}^T\bm{\alpha}+\bm{\beta}^T\bm{\beta}\big).\label{ehh}
\end{align}
	Since $\bm{\alpha}$ and $\bm{\beta}$ are i.i.d between each other we focus on the product $\bm{\alpha}^T\bm{\alpha}$ and the results are also valid for $\bm{\beta}^T\bm{\beta}$.
	Let us define $\mathbf{v}=\mathbf{R}^{-1/2}\big(\bm{\alpha}-\frac{1}{\sqrt{2}}\bm{\mu}\big)/\sigma$ which is distributed as $\mathcal{N}\big(\mathbf{0},\mathbf{I}\big)$, then
	\begin{align}
	\bm{\alpha}&=\sigma\mathbf{R}^{1/2}\mathbf{v}+\frac{1}{\sqrt{2}}\bm{\mu}\nonumber\\
	\bm{\alpha}^T\bm{\alpha}&=\big(\sigma\mathbf{R}^{1/2}\mathbf{v}+\frac{\bm{\mu}}{\sqrt{2}}\big)^T\big(\sigma\mathbf{R}^{1/2}\mathbf{v}+\frac{\bm{\mu}}{\sqrt{2}}\big)\nonumber\\
	&=\big(\mathbf{v}+\mathbf{R}^{-1/2}\frac{\bm{\mu}}{\sqrt{2}\sigma}\big)^T\sigma^2\mathbf{R}\big(\mathbf{v}+\mathbf{R}^{-1/2}\frac{\bm{\mu}}{\sqrt{2}\sigma}\big),\label{hh}
	\end{align}
	where last step comes from simple algebraic transformations. Notice that
	\begin{align}\label{RP}
	\mathbf{R}=\mathbf{B}\mathbf{\Lambda}\mathbf{B}^T,
	\end{align}
	which is the spectral decomposition of $\mathbf{R}$ \cite[Ch.21]{Harville.2008}. In \eqref{RP}, $\mathbf{\Lambda}$ is a diagonal matrix containing the eigenvalues of $\mathbf{R}$, and $\mathbf{B}$ is a matrix whose column vectors are the orthogonalized eigenvectors of $\mathbf{R}$. In order to find the eigenvalues, $\lambda$s, of $\mathbf{R}$, we require solving $\mathrm{det}\big(\mathbf{R}-\lambda\mathbf{I}\big)=0$ for $\lambda$, which is analytical intractable for a general matrix $\mathbf{R}$. However, for the special case of uniform spatial correlation with coefficient $\rho$ we are able of proceeding as follows
	
	\begin{align}
	\mathrm{det}\big(\mathbf{R}\!-\!\lambda\mathbf{I}\big)&\!=\!\mathrm{det}\left( \left[\begin{smallmatrix}
	1\!-\!\lambda& \rho & \cdots & \rho \\
	\rho & 1\!-\!\lambda & \cdots & \rho \\
	\vdots &\vdots &\ddots & \vdots \\
	\rho &\rho &\cdots & 1\!-\!\lambda
	\end{smallmatrix}\right]\right)\nonumber\\
	&=\!\mathrm{det}\left((1\!-\!\lambda\!-\!\rho)\left(\mathbf{I}\!+\! \left[\begin{smallmatrix}
	\frac{1}{1\!-\!\lambda\!-\!\rho}\\
	\frac{1}{1\!-\!\lambda\!-\!\rho}\\
	\vdots\\
	\frac{1}{1\!-\!\lambda\!-\!\rho}
	\end{smallmatrix}\right]\left[\begin{smallmatrix}
	\rho, &\rho, &\cdots, &\rho
	\end{smallmatrix}\right]\right)\right)\nonumber\\
	&\!\stackrel{(a)}{=}\!(1\!-\!\lambda\!-\!\rho)^M\!\!\left(\!\!1\!+\!\left[\begin{smallmatrix}
	\rho, &\rho, &\cdots, &\rho
	\end{smallmatrix}\right] \left[\begin{smallmatrix}
	\frac{1}{1-\lambda-\rho}\\
	\frac{1}{1-\lambda-\rho}\\
	\vdots\\
	\frac{1}{1-\lambda-\rho}
	\end{smallmatrix}\right]\! \right)\nonumber\\
	&\stackrel{(b)}{=}\!(1\!-\!\lambda\!-\!\rho)^{M-1}\big(1\!-\!\lambda\!+\!\rho(M\!-\!1)\big),\label{eq}
	\end{align}
	where $(a)$ comes from using the Matrix determinant lemma \cite{Harville.2008}, while $(b)$ follows after some algebraic manipulations. Now, two different eigenvalues are easily obtained by matching \eqref{eq} with $0$. These are $\lambda\!=\!1\!-\!\rho$ with multiplicity $M\!-\!1$ and $\lambda\!=\!1\!+\!\rho(M\!-\!1)$ with multiplicity $1$, thus
	\begin{align}
	\mathbf{\Lambda}=\mathrm{Diag}\Big[1-\rho,\cdots,1-\rho,1+(M-1)\rho\Big].\label{L}
	\end{align}

	Meanwhile, the corresponding eigenvectors, $\mathbf{e}$, satisfy $\mathbf{R}\mathbf{e}=\lambda\mathbf{e}$, thus,
	\begin{itemize}
		\item for $\lambda=1-\rho$ we have
		\begin{align}
		\left[\begin{smallmatrix}
		1 & \rho & \cdots & \rho \\
		\rho & 1& \cdots & \rho \\
		\vdots & \vdots & \ddots & \rho \\
		\rho & \rho & \cdots &1
		\end{smallmatrix}\right]\!\left[\begin{smallmatrix}
		e_1\\
		e_2\\
		\vdots\\
		e_M
		\end{smallmatrix}\right]\!&=\!(1\!-\!\rho)\left[\begin{smallmatrix}
		e_1\\
		e_2\\
		\vdots\\
		e_M
		\end{smallmatrix}\right]\nonumber\\
		\left[\begin{smallmatrix}
		e_1+\rho\sum_{i=2}^{M}e_i\\
		\rho e_1+e_2+\rho\sum_{i=3}^{M}e_i\\
		\vdots\\
		\rho\sum_{i=1}^{M-1}e_i +e_M
		\end{smallmatrix}\right]&=\left[\begin{smallmatrix}
		(1-\rho)e_1\\
		(1-\rho)e_2\\
		\vdots\\
		(1-\rho)e_M\\
		\end{smallmatrix}\right]\nonumber\\
		\left[\begin{smallmatrix}
		\rho\sum_{i=1}^{M}e_i\\
		\rho\sum_{i=1}^{M}e_i\\
		\vdots\\
		\rho\sum_{i=1}^{M}e_i\\
		\end{smallmatrix}\right]&=\left[\begin{smallmatrix}
		0\\
		0\\
		\vdots\\
		0
		\end{smallmatrix}\right]\longrightarrow 
		\sum_{i=1}^{M}e_i=0,\label{e1}
		\end{align}
		\item and for $\lambda=1+\rho(M-1)$ we have
		\begin{align}
		\left[\begin{smallmatrix}
		1 & \rho & \cdots & \rho \\
		\rho & 1& \cdots & \rho \\
		\vdots & \vdots & \ddots & \rho \\
		\rho & \rho & \cdots &1
		\end{smallmatrix}\right]\!\left[\begin{smallmatrix}
		e_1\\
		e_2\\
		\vdots\\
		e_M
		\end{smallmatrix}\right]\!\!&=\big(1+\rho(M-1)\big)\left[\begin{smallmatrix}
		e_1\\
		e_2\\
		\vdots\\
		e_M
		\end{smallmatrix}\right]\nonumber\\
		\left[\begin{smallmatrix}
		e_1+\rho\sum_{i=2}^{M}e_i\\
		\rho e_1\!+\!e_2\!+\!\rho\sum_{i=3}^{M}e_i\\
		\vdots\\
		\rho\sum_{i=1}^{M-1}e_i +e_M
		\end{smallmatrix}\right]\!\!\!&=\left[\begin{smallmatrix}
		\big(1+\rho(M-1)\big)e_1\\
		\big(1+\rho(M-1)\big)e_2\\
		\vdots\\
		\big(1+\rho(M-1)\big)e_M\\
		\end{smallmatrix}\right]\nonumber\\
		\left[\begin{smallmatrix}
		\rho\sum_{i=1}^{M}e_i\\
		\rho\sum_{i=1}^{M}e_i\\
		\vdots\\
		\rho\sum_{i=1}^{M}e_i\\
		\end{smallmatrix}\right]&=\left[\begin{smallmatrix}
		\rho M e_1\\
		\rho M e_2\\
		\vdots\\
		\rho M e_M
		\end{smallmatrix}\right]\longrightarrow  
		e_1\!=\!\cdots\!=\!e_M.\label{e2}
		\end{align}		
	\end{itemize}
Thus, the eigenvectors associated to $\lambda=1-\rho$ satisfy \eqref{e1}, while the eigenvector associated to $\lambda\!=\big(1+\rho(M-1)\big)$ satisfies \eqref{e2}. After orthogonalization by using the Gram-Schmidt process \cite{Higham.2002}, and normalization,  the resulting vectors still satisfy either \eqref{e1}  or \eqref{e2} according to the case. For the latter, the resulting vector is $[1/\sqrt{M}, 1/\sqrt{M},\cdots,1/\sqrt{M}]$, therefore
\begin{align}\label{PP}
\mathbf{B}^T=\left[\begin{smallmatrix}
b_{1,1} & b_{2,1} & \cdots & b_{M,1}\\
b_{1,2} & b_{2,2} & \cdots & b_{M,2}\\
\vdots & \vdots & \ddots & \vdots \\
b_{1,M-1} & b_{2,M-1} & \cdots & b_{M,M-1}\\
\frac{1}{\sqrt{M}} & \frac{1}{\sqrt{M}} & \cdots &\frac{1}{\sqrt{M}}
\end{smallmatrix}\right],
\end{align}
where $\sum_{i=1}^{M}b_{i,j}=0$ for $j=1,\cdots M-1$.
	
Now, substituting  \eqref{RP} into \eqref{hh} yields
\begin{align}
   \bm{\alpha}^T\!\bm{\alpha}&\!=\!\Big(\!\mathbf{v}\!+\!\big(\mathbf{B}\mathbf{\Lambda}\mathbf{B}^T\!\big)^{\!-\frac{1}{2}}\!\!\frac{\bm{\mu}}{\sqrt{2}\sigma}\!\Big)^T\!\!\!\!\sigma^2\mathbf{B}\mathbf{\Lambda}\mathbf{B}^T\!\Big(\!\mathbf{v}\!+\!\big(\mathbf{B}\mathbf{\Lambda}\mathbf{B}^T\!\big)^{\!-\frac{1}{2}}\!\!\frac{\bm{\mu}}{\sqrt{2}\sigma}\!\Big)\nonumber\\
	&\stackrel{(a)}{=}\Big(\mathbf{B}^T\!\mathbf{v}\!+\!\mathbf{\Lambda}^{\!-\frac{1}{2}}\mathbf{B}^T\!\frac{\bm{\mu}}{\sqrt{2}\sigma}\Big)^T\!\!\sigma^2\mathbf{\Lambda}\Big(\mathbf{B}^T\!\mathbf{v}\!+\!\mathbf{\Lambda}^{\!-\frac{1}{2}}\mathbf{B}^T\!\frac{\bm{\mu}}{\sqrt{2}\sigma}\Big)\nonumber\\
	&\stackrel{(b)}{=}\Big(\mathbf{Q}+\mathbf{d}\frac{1}{\sqrt{2}}\frac{\mu}{\sigma}\Big)^T\sigma^2\mathbf{\Lambda}\Big(\mathbf{Q}+\mathbf{d}\frac{1}{\sqrt{2}}\frac{\mu}{\sigma}\Big)\nonumber\\
	&\stackrel{(c)}{=}\sigma^2\sum_{i=1}^{M}\lambda_{ii}\Big(\varsigma_i+d_i\frac{1}{\sqrt{2}}\frac{\mu}{\sigma}\Big)^2,\label{hhn}
	\end{align}
	where $(a)$ comes after some algebraic transformations, $(b)$ follows from taking $\mathbf{Q}=\mathbf{B}^T\mathbf{v}\sim\mathcal{N}\big(\mathbf{0},\mathbf{I}\big)$ and $\mathbf{d}=\mathbf{\Lambda}^{-\frac{1}{2}}\mathbf{B}^T\times\mathbf{1}_{M\times 1}$, for which using \eqref{PP} and \eqref{L} yields
	\begin{align}\label{d}
	\mathbf{d}=\bigg[0, 0, \cdots,0, \sqrt{\tfrac{M}{1+(M-1)\rho}}\bigg]^T,
	\end{align}	
	and finally $(c)$ comes from setting $\varsigma_i\sim\mathcal{N}\big(0,1\big)$. Notice that $\lambda_{i}$ is the $i$-th eigenvalue of $\mathbf{R}$. Using \eqref{hhn} into \eqref{ehh} yields
	\begin{align}
	\xi^0_{\mathrm{SA}}&\!\stackrel{(a)}{=}\frac{\eta\varrho\sigma^2}{M}\sum_{i=1}^{M}\lambda_{ii}\bigg[\Big(\varsigma_i+d_i\frac{1}{\sqrt{2}}\frac{\mu}{\sigma}\Big)^2+\Big(\omega_i+d_i\frac{1}{\sqrt{2}}\frac{\mu}{\sigma}\Big)^2\bigg]\nonumber\\
	&\!\stackrel{(b)}{=}\frac{\eta\varrho\sigma^2}{M}\bigg[(1-\rho)\sum_{i=1}^{2(M-1)}\varsigma_i^2+\big(1+(M-1)\rho\big)\times\nonumber\\
	&\qquad\qquad\qquad\times\sum_{i=1}^{2}\Big(\omega_i+\sqrt{\frac{M}{2\big(1+(M\!-\!1)\rho\big)
			}}\frac{\mu}{\sigma}\Big)^2\bigg]\nonumber\\
		&\!\stackrel{(c)}{\sim} \frac{\eta\varrho}{2M(1+\kappa)}\bigg[(1-\rho)\chi^2\Big(2(M-1),0\Big) +\nonumber\\
		&\qquad\ \ +\big(1+(M-1)\rho\big)\chi^2\Big(2,\frac{2\kappa M}{1+(M-1)\rho}\Big)\bigg],\label{pdfc}
	\end{align}
	where $(a)$ comes from defining $\omega_i\sim \mathcal{N}\big(0,1\big)$ to use when evaluating the term $\bm{\beta}^T\bm{\beta}$ in \eqref{ehh}, which has the same form given in \eqref{hhn}, while $(b)$ follows from using \eqref{d}. In $(b)$ we also regrouped similar terms, which allows writing $(c)$ after using the direct definition of a non-central chi-square distribution \cite[Ch.2]{Proakis.2001} along with \eqref{kap} and \eqref{sigm2mu}.
	
	Under the assumption of uniform spatial correlation with parameter $\rho$ ($r_{i,j}=\rho,\ \forall i\ne j$), \eqref{pdfc} holds; however, notice that for the $\mathrm{AA}$ scheme we were able of writing the distribution merely as a function of parameter $\delta$, which is not linked to any specific kind of correlation; hence, we can expect that the behavior under the $\mathrm{SA}$ scheme depends, at least approximately, on $\delta$ rather on the specific entries of matrix $\mathrm{R}$. To explore this, we substitute $\rho=\frac{\delta-M}{M(M-1)}$ coming from \eqref{delta2}, into \eqref{pdfc}, such that we attain
		\begin{align}
		\xi^0_{\mathrm{SA}}\!\sim\! \frac{\eta\varrho}{2M^2(1\!+\!\kappa)}\!\bigg[\!\frac{M^2\!-\!\delta}{M\!-\!1}\chi^2\!\Big(\!2(M\!-\!1),0\!\Big)\! +\!\delta\chi^2\!\Big(\!2,\frac{2\kappa M^2}{\delta}\Big)\!\bigg].\label{pdfc_app}
		\end{align}
        For validation purposes let us assume an exponential correlation matrix such that $r_{i,j}=\tau^{|i-j|}$ where $\tau$ is the correlation coefficient of neighboring antennas\footnote{This model is physically reasonable since the correlation decreases with increasing distance between antennas \cite{Chen.2015}. Still, simulation results that are not included in this paper evidence that  the trends and accuracy of our results under such exponential correlation model remain valid independently of the correlation profile.}, hence $\delta\!=\!M\!+\!2\sum_{i=1}^{M\!-\!1}\!(M\!-\!i)\tau^i$. Fig.~\ref{Fig2} shows Monte Carlo results comparing the exact \eqref{ehh} and approximate \eqref{pdfc_app} expressions of the harvested energy distribution under the $\mathrm{SA}$ scheme.
        As shown, they match accurately except when the setup is jointly characterized  by  very small $\kappa$ and relatively large $M$ and $\tau$, as it is the case of $\kappa=0,\ M=16$ and $\tau=0.8$. In the other cases, \eqref{pdfc_app} holds accurately, while matching exactly \eqref{ehh} and \eqref{pdfc} under uniform spatial correlation. Notice that when $\tau=0$ or $1$, the system is under an extreme case of uniform spatial correlation and \eqref{pdfc_app} becomes exact; hence, for fixed $\kappa$ and $M$,  \eqref{pdfc_app} has the least accurate convergence to the exact distribution for some $0<\tau^*<1$. The value(s) of $\tau^*$ could be seen as the one(s) providing the greatest difference or distance between $\mathrm{R}$ and its equivalent uniform correlation matrix with coefficient $\rho$.\footnote{There are many similarity metrics and concepts of distance between matrices in the literature, such as the $\ell_p-$norm distances, the trace distance, the correlation matrix distance \cite{Herdin.2005}, just to name a few. Each of them has been shown to be appropriate under different goals. However, in the context of our work it is not clear which one fits better for finding $\tau^*$, and even having such information, the following analyses are expected to be cumbersome. For this reason, we have carried out extensive simulations and found that $\tau^*$ is unique and increases very slowly with $M$ such that for $M=3\rightarrow\tau^*\approx 0.75$ while for $M=32\rightarrow\tau\approx 0.85$. Therefore, notice that our results are expected to be accurate in most of practical scenarios where also a non-negligible LOS parameter influences positively on the accuracy of \eqref{pdfc_app}, as shown in Fig.~\ref{Fig2}b.}
	\begin{figure}[t!]
		\centering
		\subfigure{\includegraphics[width=0.95\columnwidth]{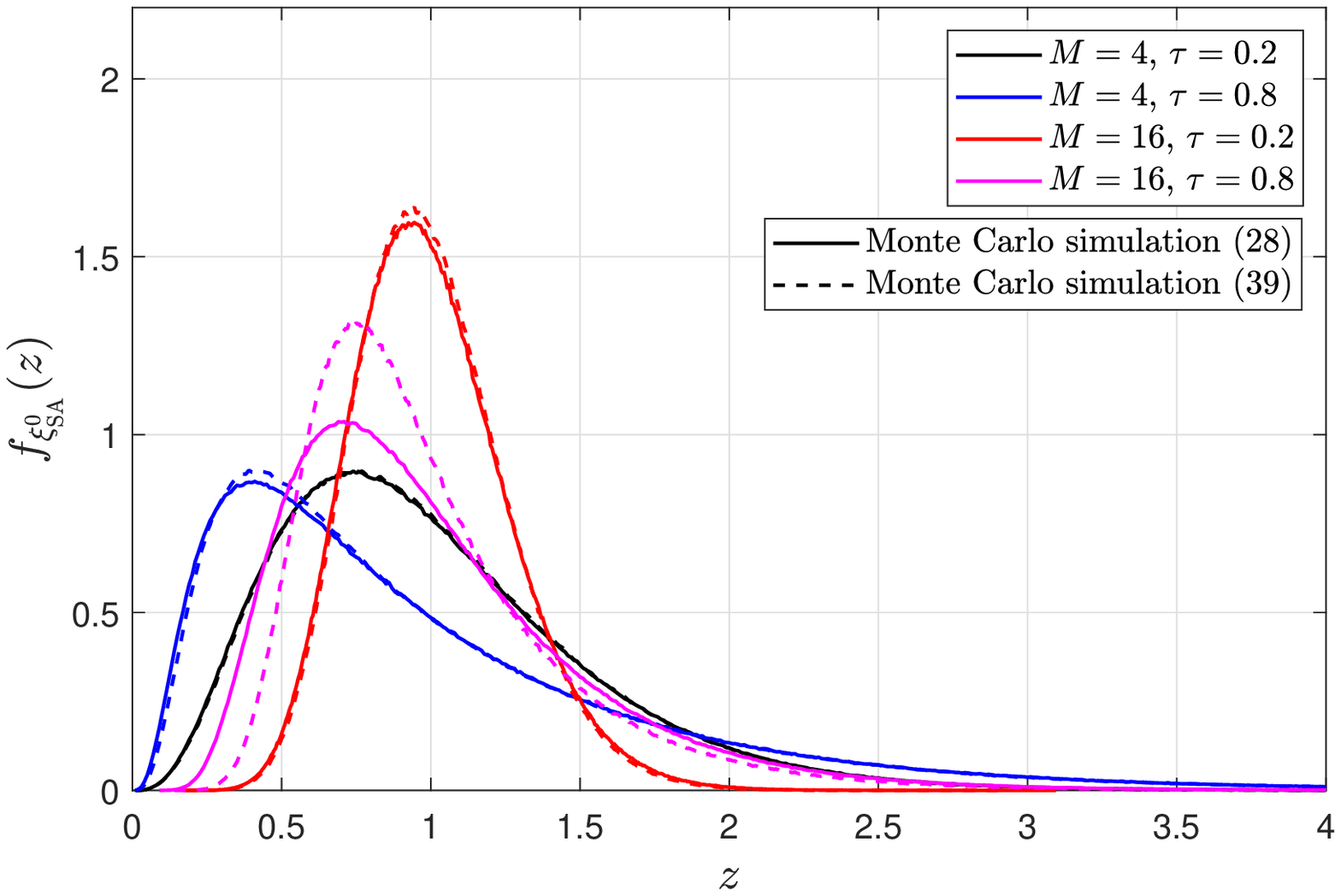}}\\ 
		\vspace{1mm}
		\subfigure{\includegraphics[width=0.95\columnwidth]{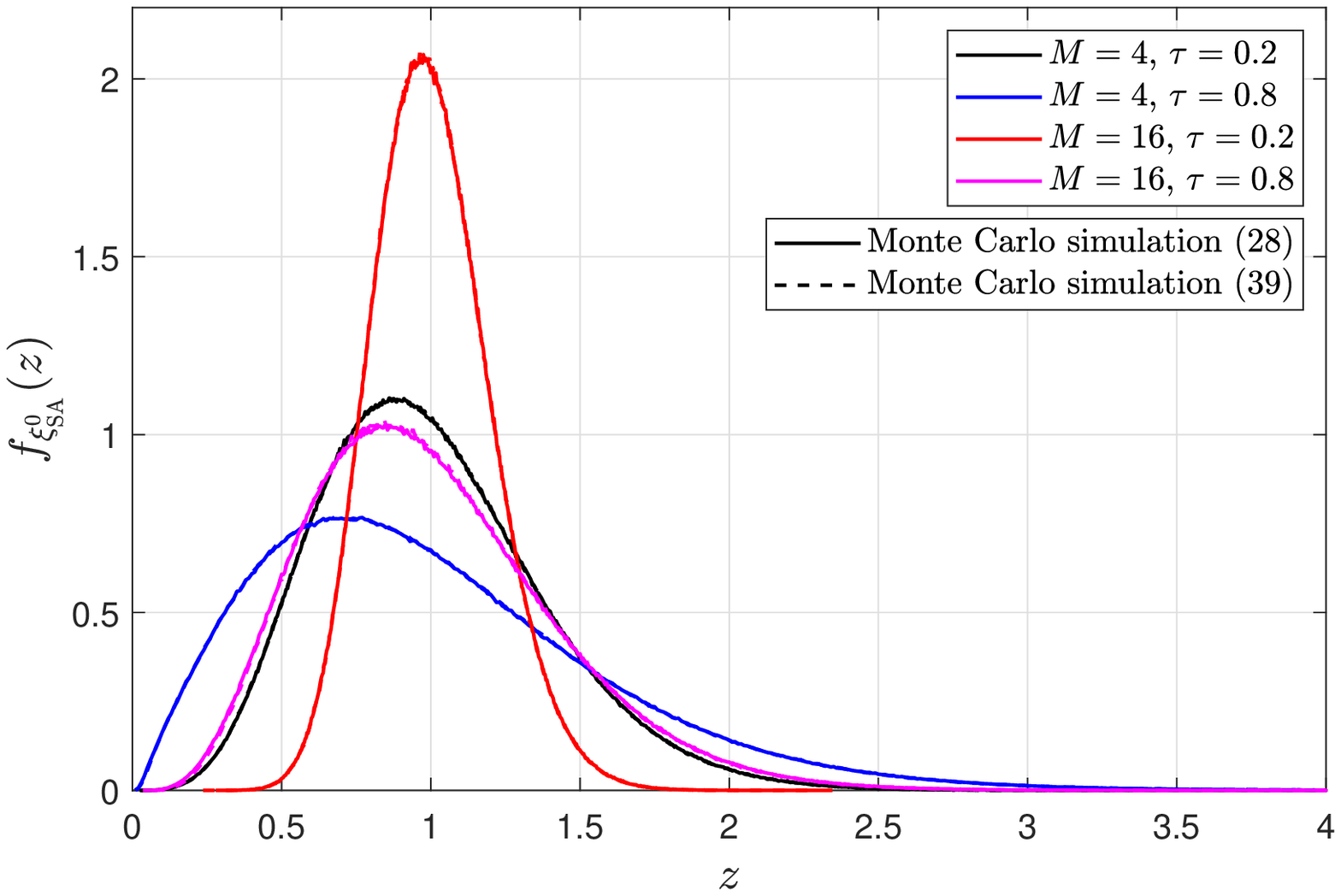}}\\
		\caption{PDF of $\xi_{_\mathrm{SA}}^0$ based on Monte Carlo simulation of exact \eqref{ehh} and approximate \eqref{pdfc_app} expressions. We set $\eta\varrho=1$ and use $M\in\{4,16\}$ and  $\tau\in\{0.2,0.8\}$; while $(a)$ $\kappa=0$ (top), and $(b)$ $\kappa=3$ (bottom).}
		\label{Fig2}
	\end{figure}		
   \begin{remark}\label{re2}
    	Therefore, $\xi^0_{\mathrm{SA}}$  is approximately distributed as a linear combination of a chi-square RV and a non-central chi-square RV, with $2(M-1)$ and $2$ degrees of freedom, respectively. Unfortunately, it seems intractable finding a closed-form expression for the distribution of $\xi^0_{\mathrm{SA}}$ even when using \eqref{pdfc_app}, except for
    	\begin{itemize}
    		\item $\delta=0$, for which $\xi^0_{\mathrm{SA}}\sim \frac{\eta\varrho}{2(M\!-\!1)(1\!+\!\kappa)}\chi^2\big(2(M\!-\!1),0\big)\!+\!\frac{\eta\varrho\kappa}{1\!+\!\kappa}$. The latter term results from the fact that when $\delta\rightarrow 0$ the PDF of $\delta\chi^2\big(2,\frac{2\kappa M^2}{\delta}\big)$  is 1 at $2M\kappa$, and $0$ otherwise;
    		\item $\delta=M$, for which $\xi^0_{\mathrm{SA}}\sim \frac{\eta\varrho}{2M(1+\kappa)}\Big[\chi^2\big(2(M-1),0\big)+\chi^2\big(2,2M\kappa\big)\Big]\sim \frac{\eta\varrho}{2M(1+\kappa)} \chi^2\big(2M,2M\kappa\big)$, which can be easily verified by using the direct definition of a non-central chi-squared RV;    		
    		\item $\delta=M^2$, for which $\xi^0_{\mathrm{SA}}\sim \frac{\eta\varrho}{2(1+\kappa)}\chi^2\big(2,2\kappa\big)$.
    	\end{itemize}
    \end{remark}
    \begin{remark}\label{re2_5}
    	For full positive correlation, e.g., $\delta=M^2$, the performance of the $\mathrm{SA}$ scheme matches that of the $\mathrm{OA}$. This is an expected result since even by switching antennas the energy harvested at $S$ keeps the same.
    \end{remark}
	\subsubsection{$\mathrm{OA-CSI}$}
	According to \eqref{1acsiv2} finding the distribution of $\xi^0_{\mathrm{OA-CSI}}$ is equivalent to the problem of finding the distribution of the Signal-to-Noise Ratio in a correlated Rician single-input multiple-output (SIMO) channel scenario, where the receiver with $M$ antennas uses Selection Combining (SC). We first analyze the case of uniform spatial correlation with coefficient $\rho$, for which we can directly use \cite[Eq.(21)]{Chen.2004} to state
	\vspace{-2mm}
	\begin{align}
	F_{\xi^0_{\mathrm{OA-CSI}}}(\eta\varrho x)&=\!e^{-\frac{\kappa}{\rho}}\!\int\limits_{0}^{\infty}\!\bigg[1\!-\!\mathcal{Q}\bigg(\sqrt{\frac{2\rho t}{1\!-\!\rho}},\sqrt{\frac{2(1\!+\!\kappa)x}{1-\rho}}\bigg)\bigg]^M\!\!\!\times\nonumber\\
	&\qquad\qquad\qquad\ \ \times e^{-t}I_0\Big(2\sqrt{\frac{\kappa t}{\rho}}\Big)\mathrm{d}t,\label{Fz0}\\
	f_{\xi^0_{\mathrm{OA-CSI}}}(x)&=\!\frac{d}{dx}F_{\xi^0_{\mathrm{OA-CSI}}}(x)\!=\!\frac{M(1\!+\!\kappa)}{\eta\varrho(1\!-\!\rho)}e^{-\frac{\kappa}{\rho}\!-\!\frac{1+\kappa}{\eta\varrho(1\!-\!\rho)}x}\!\times\nonumber\\
	\int\limits_{0}^{\infty}e^{-\frac{t}{1-\rho}}I_0&\Big(2\sqrt{\frac{\kappa t}{\rho}}\Big) \!\bigg[1\!-\!\mathcal{Q}\Big(\sqrt{\frac{2\rho t}{1\!-\!\rho}},\sqrt{\frac{2(1\!+\!\kappa)x}{\eta\varrho(1\!-\!\rho)}}\Big)\!\bigg]^{M\!-\!1}\!\!\!\!\times\nonumber\\
	&\qquad\qquad\times I_0\Big(\frac{2}{1\!-\!\rho}\sqrt{\frac{\rho t (1\!+\!\kappa)x}{\eta\varrho}}\Big)\mathrm{d}t.\label{fz0}
	\end{align}
	Now and as done in the previous subsection, we utilize the transformation $\rho=\frac{\delta-M}{M(M-1)}$, which hopefully allows converting \eqref{Fz0} and \eqref{fz0} into expressions valid  for systems with any kind of correlation. By doing so, the PDF for instance becomes the one shown in \eqref{fz} at the top of the next page.
	\begin{figure}[t!]
		\centering
		\subfigure{\includegraphics[width=0.95\columnwidth]{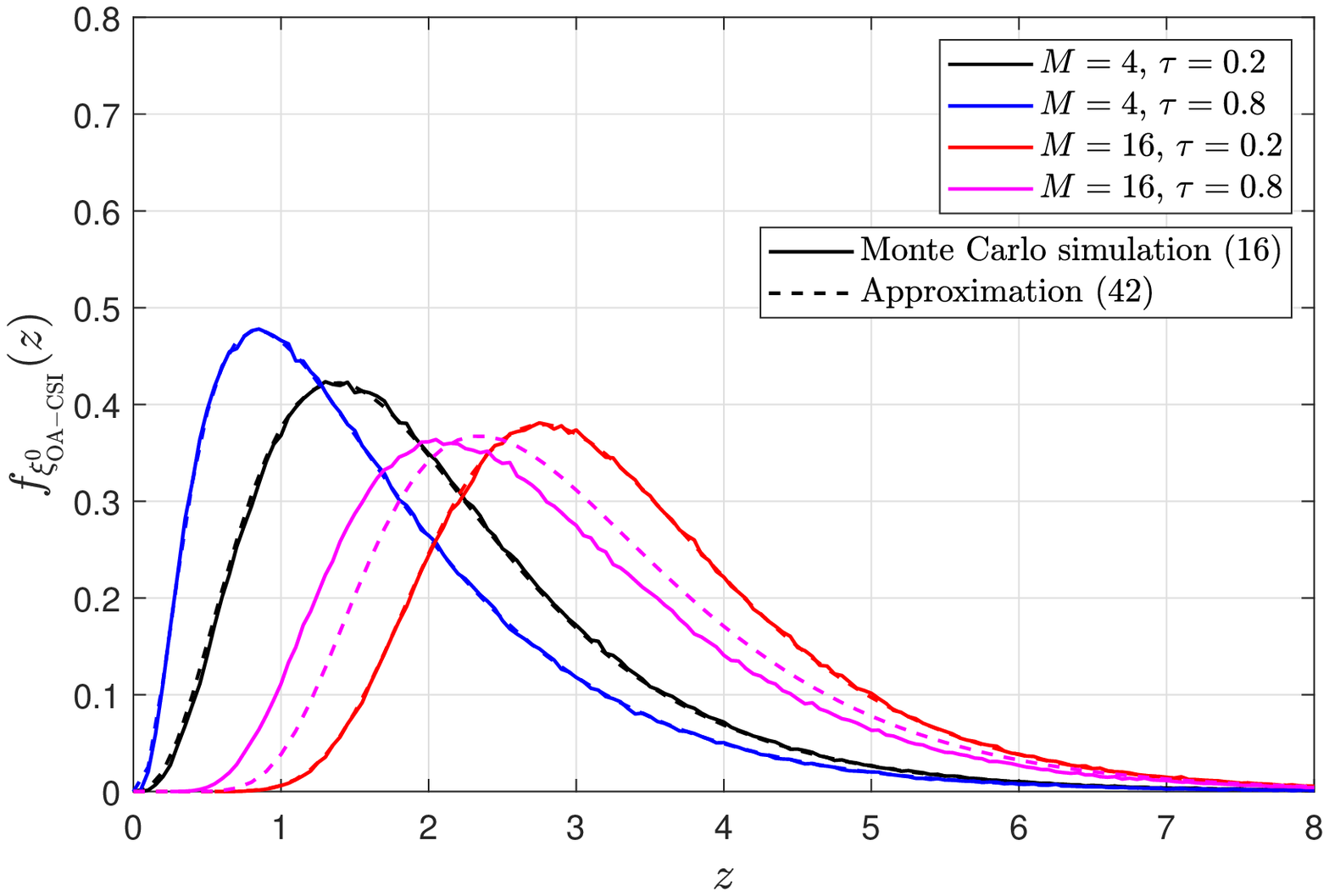}}\\
		\vspace{1mm}
		\subfigure{\includegraphics[width=0.95\columnwidth]{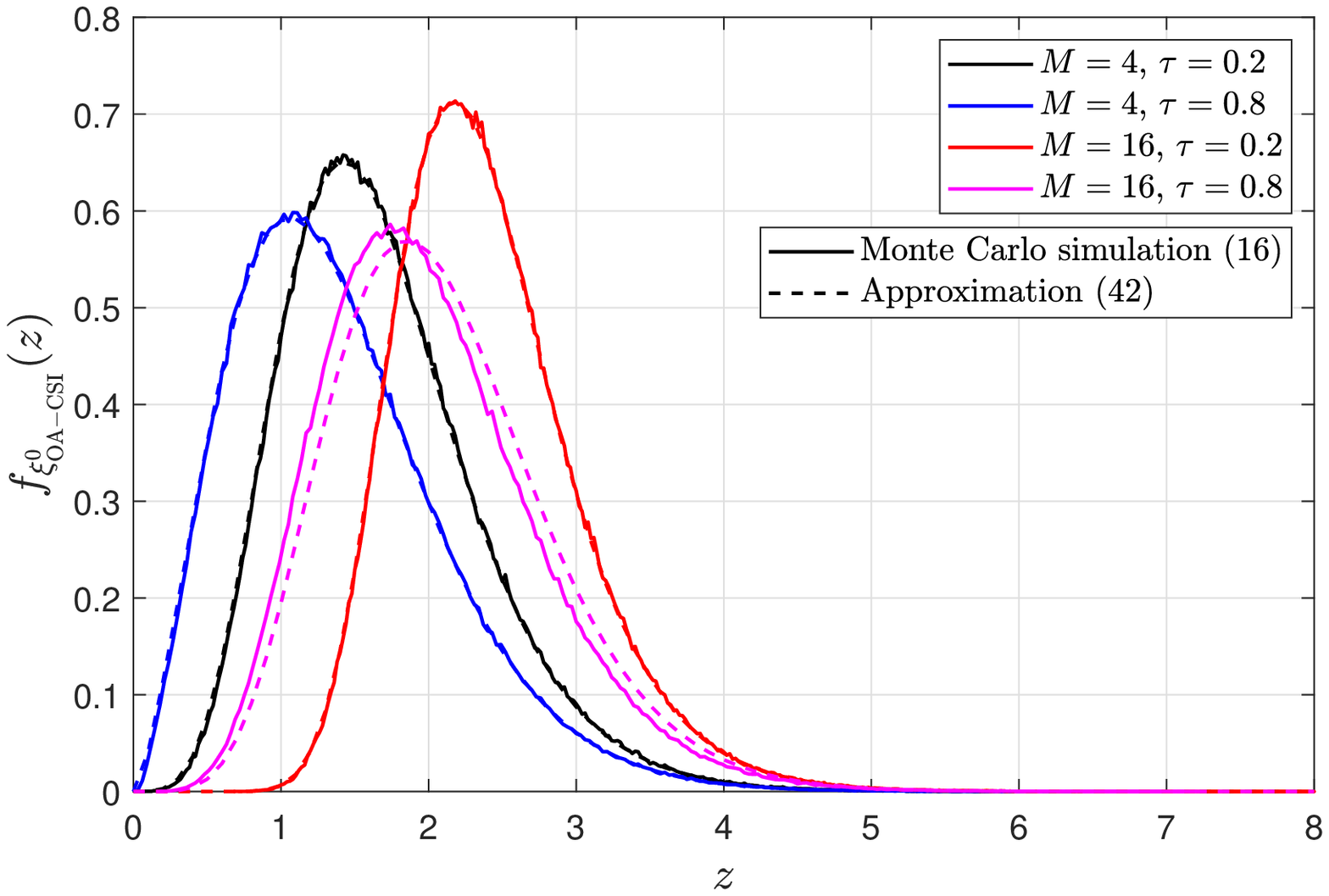}}
		\caption{PDF of $\xi_{_\mathrm{OA-CSI}}^0$ based on Monte Carlo simulation of \eqref{1acsiv2} and numerical evaluation of approximation given in \eqref{fz}. We set $\eta\varrho=1$ and use $M\in\{4,16\}$ and  $\tau\in\{0.2,0.8\}$; while $(a)$ $\kappa=0$ (top), and $(b)$ $\kappa=3$ (bottom).}
		\label{Fig3}
	\end{figure}
	\begin{figure*}
		\small
		\begin{align}
		f_{\xi^0_{\mathrm{OA-CSI}}}(x)=&\frac{M^2(M-1)(1+\kappa)}{\eta\varrho(M^2-\delta)}e^{-\frac{\kappa M(M-1)}{\delta-M}-\frac{(1+\kappa)M(M-1)}{\eta\varrho(M^2-\delta)}x}\int\limits_{0}^{\infty}e^{-\frac{tM(M-1)}{M^2-\delta}}I_0\Big(2\sqrt{\frac{\kappa M(M-1) t}{\delta-M}}\Big)\times\nonumber\\
		&\times \bigg[1-\mathcal{Q}\Big(\sqrt{\frac{2(\delta-M) t}{M^2-\delta}},\sqrt{\frac{2(1+\kappa)M(M-1)x}{\eta\varrho(M^2-\delta)}}\Big)\bigg]^{M-1}I_0\bigg(\frac{2}{M^2-\delta}\sqrt{\frac{(\delta-M)M(M-1) t (1+\kappa)x}{\eta\varrho}}\bigg)\mathrm{d}t. \label{fz}
		\end{align}
		\hrule
	\end{figure*}	

    Similar to our approach in the previous subsection, herein we validate \eqref{fz} by comparing it with the distribution of the harvested energy using Monte Carlo simulation in scenarios with exponential correlated fading. Results in Fig.~\ref{Fig3} evidence the same behavior as discussed in Fig.~\ref{Fig2} and corroborate the accuracy of \eqref{fz}, which is only affected when a very small $\kappa$, and relatively large $M$ and $\tau$, coincide together.
	
    For two specific setups it is even possible to get simplified expressions as follows
	\begin{itemize}
		\item For uncorrelated channels, for which $\delta=M$, the channel coefficients are i.i.d, thus,
	\begin{align}	
	F_{\xi^0_{\mathrm{OA-CSI}}}(\eta\varrho x)&=\mathbb{P}\big[\xi^0_{\mathrm{OA-CSI}}<\eta\varrho x\big]\nonumber\\
	&=\mathbb{P}\big[\max_{i=1,...,M}|h_{i}|^2<x\big]\nonumber\\
	&\stackrel{(a)}{=}\mathbb{P}\big[|h_{i}|^2<x\big]^M\nonumber\\
	&\stackrel{(b)}{=}\mathbb{P}\big[z<2(1+\kappa)x\big]^M\nonumber\\
	&=F_Z\big(2(1+\kappa)x\big)^M,
	\end{align}
	where in $(a)$ we take advantage of the i.i.d property of the channel realizations and $(b)$  comes from making $z=2(1+\kappa)|h_{i}|^2$ where $Z\sim \chi^2(2,2\kappa)$. Now,
	\begin{align}\label{fe}
	&f_{\xi^0_{\mathrm{OA-CSI}}}(x)=\frac{d}{d x}F_{\xi^0_{\mathrm{OA-CSI}}}(x)=\frac{d}{dx}F_Z\Big(\frac{2(1+\kappa)x}{\eta\varrho}\Big)^M\nonumber\\
	&\ \ \stackrel{(a)}{=}\frac{2(1\!+\!\kappa)M}{\eta\varrho}F_Z\!\Big(\!\frac{2(1\!+\!\kappa)}{\eta\varrho}x\!\Big)^{M\!-\!1}\!\!\!f_Z\Big(\!\frac{2(1\!+\!\kappa)}{\eta\varrho}x\!\Big),
	\end{align}
	where $(a)$ comes from using the chain rule. Notice that $f_Z(x)$ and $F_Z(x)$ are given respectively in \eqref{fZ} and \eqref{FZ} with $\varphi=2$ and $\psi=2\kappa$.
	
	\item For fully correlated channels, for which $\delta=M^2$, the performance under the $\mathrm{OA-CSI}$ scheme matches that of the $\mathrm{OA}$ strategy since the fading behaves instantaneously equal over all the antennas.
	
	Finally, notice that when correlation is the minimum, $\delta=0$, the performance gain of this scheme over that of the $\mathrm{OA}$ strategy should be the maximum.
\end{itemize}
\subsubsection{$\mathrm{AA-CSI}$}
 For the $\mathrm{AA-CSI}$ scheme the characterization is much easier since $\xi^0_{\mathrm{AA-CSI}}=M\xi^0_{\mathrm{SA}}$ according to \eqref{ME}, thus, by using \eqref{pdfc_app} we attain
	\begin{align}\label{aaacsir}
	&\xi^0_{\mathrm{AA-CSI}}\nonumber\\
	&\sim\frac{\eta\varrho}{2M(1\!+\!\kappa)}\bigg[\frac{M^2\!-\!\delta}{M\!-\!1}\!\chi^2\!\Big(\!2(M\!-\!1)\!,0\!\Big)\!+\!\delta\chi^2\!\Big(\!2,\!\frac{2\kappa M^2}{\delta}\Big)\!\bigg].
	\end{align}
	Notice that the analysis in Remark~\ref{re2} can be extended to this case straightforwardly.
	
\subsubsection{Comparisons and Remarks}
Since the distribution of the harvested energy is related in all the cases with the non-central chi-squared distribution but for the $\mathrm{OA\!-\!CSI}$ scheme, we are able to find the mean and variance statistics according to \eqref{meanvariance}. Under uncorrelated channels, $\delta=M$, we can also provide an upper bound for the mean of $\xi^0_{1_\mathrm{OA-CSI}}$ by using \cite[The. 2.1]{Aven.1985} such that
	\begin{align}
	&\frac{\mathbb{E}\big[\xi^0_{\mathrm{OA-CSI}}\big]}{\eta\varrho}\le\frac{1}{M}\sum_{i=1}^{M}\mathbb{E}[|h_{i}|^2]+\sqrt{\frac{M-1}{M}}\times\nonumber\\
	&\ \ \ \times\sqrt{\sum_{i=1}^{M}\!\mathrm{VAR}\big[|h_{i}|^2\big]\!+\!\sum_{i=1}^{M}\!\mathbb{E}\big[|h_{i}|^2\big]^2\!-\!\frac{1}{M}\Big(\!\sum_{i=1}^{M}\!\mathbb{E}\big[|h_{i}|^2\big]\Big)^2}\nonumber\\
	&\mathbb{E}\big[\xi^0_{\mathrm{OA-CSI}}\big]\stackrel{(a)}{\le}\eta\varrho\Big(\mathbb{E}[|h_{i}|^2]+\sqrt{(M-1)\mathrm{VAR}[|h_{i}|^2]}\Big)\nonumber\\
	&\qquad\qquad\ \ \ \stackrel{(b)}{=}\eta\varrho\Big(1+\frac{\sqrt{(1+2\kappa)(M-1)}}{1+\kappa}\Big).\label{upp}
	\end{align}
	where $(a)$ comes from using $\sum_{i=1}^{M}\!\mathbb{E}\big[\!|h_{i}|^2\big]\!=\!M\mathbb{E}\big[\!|h_{i}|^2\big]$ and $\sum_{i=1}^{M}\!\mathrm{VAR}\big[\!|h_{i}|^2\big]\!=\!M\mathrm{VAR}\big[\!|h_{i}|^2\big]$, while $(b)$ follows after using \eqref{meanvariance}. Notice that for $M=1$ the equality is satisfied as expected. Regarding the variance, it is shown in \cite[The. 3.1]{Papadatos.1995} that the best  upper bound for $\mathrm{VAR}\big[\max\limits_{i=1,...,M}|h_{i}|^2\big]$ is $M\mathrm{VAR}\big[|h_{1}|^2\big]$. However, that result is not tight for the fading distribution of our interest here, and we dispense with that result.

In Table~\ref{table} we summarize the statistics of the energy harvested under the operation of each of the WET schemes. We also include the CDF expressions, which can be easily obtained by using $F_Z(z)$ in \eqref{FZ} in most cases. When it was impossible a full characterization for all the values of $\delta$, we provide the results for the case of i.i.d channel realizations where $\delta=M$.
 Some remarks follow next.
 \begin{remark}\label{re5}
 	 For the $\mathrm{AA}$ scheme,	both the mean and variance of the harvested energy, increase with the correlation parameter $\delta$. Thus, the greatest gain in the mean harvested energy occurs when channels are fully positive correlated, for which $\delta=M^2$. 
 		In fact, when channels are fully positive correlated $S$ harvests $M$ times more energy under the $\mathrm{AA}$ scheme than when using the $\mathrm{OA}$ strategy, but with a dispersion $M^2$ times greater. 
 		This can be  easily corroborated also from the very definition of $\xi^0_{\mathrm{AA}}$ since $\xi^0_{\mathrm{AA}}\!=\!\frac{\eta\varrho}{M}\Big|\sum_{i=1}^{M}h_{i}\Big|^2\!\stackrel{(a)}{=}\!\frac{\eta\varrho}{M}\Big|Mh_{i}\Big|^2\!=\!\eta\varrho M\big|h_{i}\big|^2\!=\!\eta\varrho M\xi^0_{\mathrm{OA}}$, where $(a)$ comes from the fact that for $\delta=M^2$ all the channel coefficients are the same, e.g., $h_{i}=h_{n},\ \forall i,n=1,...,M$.
 		 	
 		For the particular value $\delta=0$, the dispersion of the harvested energy becomes $0$, thus in such situations $T$ provides a stable, non-random, energy supply to $S$, such that $\xi^0_{\mathrm{AA}}=\frac{\eta\varrho M\kappa}{1+\kappa}$. 
 		This setup guarantees the exact prediction of the harvested energy, which  grows linearly with $M$ in LOS channels. However, under Rayleigh conditions the harvested energy becomes $0$, which is because of cancellation of the multiple path signals at the sensor.
 \end{remark}
\begin{remark}\label{re3}
		The average energy harvested under the $\mathrm{SA}$ and $\mathrm{AA-CSI}$ schemes is constant, while the variance is a convex function of the spatial correlation parameter $\delta$, which can be easily corroborated by checking that the second derivative of $\mathrm{VAR}\big[\xi^0_{\mathrm{SA}}\big]$ and $\mathrm{VAR}\big[\xi^0_{\mathrm{AA-CSI}}\big]$ with respect to $\delta$ is non negative.
		The minimum variance for a given LOS parameter occurs when $\delta=M\big(1-\min\big(\kappa(M-1),1\big)\big)$, which comes from setting the first derivative of $\mathrm{VAR}\big[\xi^0_{\mathrm{SA}}\big]$ or $\mathrm{VAR}\big[\xi^0_{\mathrm{AA-CSI}}\big]$ to $0$ and solving for $\delta$, while using also the restriction $\delta\ge 0$. Therefore, for increasing/decreasing $\delta$ above/below $M\big(1-\min\big(\kappa(M-1),1\big)\big)$, the variance increases. 
				
		Additionally, notice that the variance of the harvested energy under the $\mathrm{SA}$ scheme is
		\begin{itemize}
			\item $\frac{\eta^2\varrho^2(1+2\kappa)}{M(1+\kappa)^2}$ for $\delta=M$, which is inversely proportional to the number of antennas;
			\item $\frac{\eta^2\varrho^2(1+2\kappa)}{(1+\kappa)^2}$ for $\delta=M^2$, which is independent of the number of antennas;
		\end{itemize}	
	\end{remark}
	\begin{table*}[!t]
		\centering
		\caption{Statistics of the Harvested Energy for $|\mathcal{S}|=1$ and under the Ideal Linear EH Model}
		 \label{table}
		\begin{tabular}{P{1.2cm}|P{3cm}|P{3.8cm}|P{7.6cm}}			
			\toprule
			\normalsize{Scheme} & \normalsize{$\mathbb{E}\big[\ \!\xi^0\!\ \big]/(\eta\varrho)$} & \normalsize{$\mathrm{VAR}\big[\ \!\xi^0\!\ \big]/(\eta^2\varrho^2)$} & \normalsize{$F_{\xi^0}(\eta\varrho z)$} \\ 
			\bottomrule
			$\mathrm{OA}$  & $1$ & $ \frac{1+2\kappa}{(1+\kappa)^2}$ & $1-\mathcal{Q}_1\big(\sqrt{2\kappa},\sqrt{2(1+\kappa)z}\big)$\\[8pt] 
			$\mathrm{AA}$  & $\frac{\delta+\kappa M^2}{M(1+\kappa)}$ & $\frac{\delta(\delta+2\kappa M^2)}{M^2(1+\kappa)^2}$ & $1-\mathcal{Q}_1\Big(M\sqrt{\frac{2\kappa}{\delta}},\sqrt{\frac{2M(1+\kappa)z}{\delta}}\Big)$\\[8pt]
			$\mathrm{SA}$  & $1$ & $\frac{(M+2\kappa\delta)M^2-2\delta M(1+\kappa)+\delta^2}{M^3(M-1)(1+\kappa)^2}$ & $\delta=M\Longrightarrow$ $1-\mathcal{Q}_{M}\big(\sqrt{2M\kappa},\sqrt{2M(1+\kappa)z}\big)$\\ \hline
			$\mathrm{OA\!-\!CSI}$ &
			$\int\limits_{0}^{\infty}\frac{x}{\eta\varrho}f_{\xi^0_{\mathrm{OA-CSI}}}(x)\mathrm{d}x$  & $\frac{1}{\eta^2\varrho^2}\bigg(\int\limits_{0}^{\infty}\!x^2f_{\xi^0_{\mathrm{OA\!-\!CSI}}}\!\!(x)\mathrm{d}x-$\qquad & $e^{\!-\frac{\kappa M(M\!-\!1)}{\delta-M}}\!\int\limits_{0}^{\infty}\bigg[1\!-\!\mathcal{Q}\Big(\sqrt{\frac{2(\delta\!-\!M) t}{M^2\!-\!\delta}},\sqrt{\frac{2(1\!+\!\kappa)M(M\!-\!1)z}{M^2-\delta}}\Big)\bigg]^M\!\!\!e^{\!-t}\times$\\
			& $\!\stackrel{\delta=M}{\le}\! 1\!+\!\frac{\sqrt{(1\!+\!2\kappa)(M\!-\!1)}}{1+\kappa}$ & \qquad\qquad\ \  $\mathbb{E}\big[\xi^0_{\mathrm{OA-CSI}}\big]^2\bigg)$ &  $\times I_0\Big(2\sqrt{\frac{\kappa M(M\!-\!1) t}{\delta-M}}\Big)\mathrm{d}t\!\stackrel{\delta=M}{=}\!\Big[1\!-\!\mathcal{Q}_1\big(\sqrt{2\kappa},\sqrt{2(1\!+\!\kappa)z}\big)\Big]^M$\\ \cline{2-4}
			$\mathrm{AA\!-\!CSI}$ & $M$& $\frac{(M+2\kappa\delta)M^2-2\delta M(1+\kappa)+\delta^2}{M(M-1)(1+\kappa)^2}$ & $\delta=M\Longrightarrow$ $1-\mathcal{Q}_M\big(\sqrt{2M\kappa},\sqrt{2(1+\kappa)z}\big)$\\
			\bottomrule
		\end{tabular}
	\end{table*}
\begin{remark}
Notice that $\mathbb{E}[\xi^0_{\mathrm{OA}}]\!\le\! \mathbb{E}[\xi^0_{\mathrm{OA-CSI}}]\!\le\! \eta\varrho M	\!=\!\mathbb{E}\big[\xi^0_{\mathrm{AA-CSI}}\big]\!=\!M\mathbb{E}\big[\xi^0_{\mathrm{SA}}\big]$, since for $M\ge 1$ and $\kappa\ge 0$ we have that
\begin{align}
(1+\kappa)^2&\ge 1+2\kappa\nonumber\\
\frac{1+2\kappa}{(1+\kappa)^2}&\le 1\nonumber\\
\frac{1+2\kappa}{(1+\kappa)^2}(M-1)&\le M-1\nonumber\\
\frac{\sqrt{(1+2\kappa)(M-1)}}{1+\kappa}&\le \sqrt{M-1}\le M-1\nonumber\\
\frac{\mathbb{E}\big[\xi^0_{\mathrm{OA-CSI}}\big]}{\eta\varrho}\!=\!1\!+\!\frac{\sqrt{(1+2\kappa)(M-1)}}{1+\kappa}&\le M,\nonumber
\end{align}
which agrees with \eqref{eq1} of Theorem~\ref{the1}. We have ignored the impact of $\delta$ in the exact expression of $\mathbb{E}\big[\xi^0_{\mathrm{OA-CSI}}\big]$ since there is no closed form, instead we used the upper bound given in \eqref{upp} for the case of $\delta=M$. Finally,  $\frac{\mathbb{E}\big[\xi^0_{\mathrm{AA}}\!\big]}{\eta\varrho}\!=\!\frac{\delta+\kappa M^2}{M(1+\kappa)}\!\le\! M\!=\!\frac{\mathbb{E}\big[\xi^0_{\mathrm{AA-CSI}}\!\big]}{\eta\varrho}\!=\!M\frac{\mathbb{E}\big[\xi^0_{\mathrm{SA}}\!\big]}{\eta\varrho}$ holds since $M\ge 1$, and agrees with \eqref{eq2}.
\end{remark}
\begin{remark}\label{me-var}
	For $\kappa\ge 0$, the average harvested energy under the operation of the $\mathrm{AA}$ scheme is an increasing function of $\kappa$ since $\frac{d}{d\kappa}\mathbb{E}\big[\xi^0_{\mathrm{AA}}\big]> 0$. For $\mathrm{OA}$, $\mathrm{SA}$ and $\mathrm{AA-CSI}$ schemes, the average harvested energy does not depend on the LOS parameter. 	
	Meanwhile, the variance of the harvested energy decreases with $\kappa$ when $\mathrm{OA}$ is used, while it has a maximum on	
	$\kappa_1^*=1-\frac{\delta}{M^2}$, 
	$\kappa_2^*=\frac{(\delta-M)(M^2-\delta)}{\delta M(M-1)}$,  
	 when $\mathrm{AA}$ and, $\mathrm{SA}$ or $\mathrm{AA-CSI}$, are utilized, respectively. This is because
	 \begin{align}
	 \frac{d}{d\kappa}\mathrm{VAR}\big[\xi^0_{\mathrm{OA-CSI}}\big]&<0,\nonumber\\
	 \frac{d}{d\kappa}\mathrm{VAR}\big[\xi^0_{\mathrm{AA}}\big]&\stackequal{\kappa<\kappa_1^*}{\kappa>\kappa_1^*}0,\nonumber\\
	 \frac{d}{d\kappa}\mathrm{VAR}\big[\xi^0_{\mathrm{SA}}\big]&\stackequal{\kappa<\kappa_2^*}{\kappa>\kappa_2^*}0,\nonumber\\
	 \frac{d}{d\kappa}\mathrm{VAR}\big[\xi^0_{\mathrm{AA-CSI}}\big]&\stackequal{\kappa<\kappa_2^*}{\kappa>\kappa_2^*}0.\nonumber
	 \end{align}
\end{remark}
\begin{remark}
	The average harvested energy under the CSI-based WET schemes overcomes the CSI-free counterparts. But, interestingly, for i.i.d channels and $M\ge 2$, $T$ manages to harvest more energy on average under  the $\mathrm{AA}$ scheme than when using the $\mathrm{OA-CSI}$ scheme, even without considering the energy cost associated to the CSI acquisition, when
	$M> \frac{1+2\kappa}{\kappa^2}$. 
	This is because under such condition and $\delta=M$, $\frac{\mathbb{E}\big[\xi^0_{\mathrm{AA}}\big]}{\eta\varrho}=\frac{1+M\kappa}{1+\kappa}> 1+\frac{\sqrt{(1+2\kappa)(M-1)}}{1+\kappa}\ge \frac{\mathbb{E}\big[\xi^0_{\mathrm{OA-CSI}}\big]}{\eta\varrho}$ holds. Previous condition can also be written as 
	$\kappa>\frac{1+\sqrt{M+1}}{M}$.
\end{remark}
In many cases not only a high average of harvested energy is desirable but also a low variance. Let us assume for instance that the sensor requires a reliable energy supply each time that it harvests energy from $T$ transmissions, and that exceeds the value $\xi_{\mathrm{th}}$. The value of $\xi_{\mathrm{th}}$ may be seen as the minimum amount of energy for which $S$ can operate, which in general satisfies $\xi_{\mathrm{th}}\ge \eta\varpi_1$. In such cases a question arises: \textit{What is the more suitable WET strategy?} The probability of energy outage is the performance metric to be evaluated in this case, and it is defined as $\mathbb{P}[\xi^0<\eta\xi_{\mathrm{th}}]=F_{\xi^0}(\eta\xi_{\mathrm{th}})$.	The more reliable scheme would have the lowest probability of energy outage.
	\begin{remark}\label{re8}
	The $\mathrm{SA}$ strategy provides WET rounds with the lowest variance in the harvested energy, thus, it is more predictable. Additionally, it is the only strategy for which the variance of the harvested energy over multiple rounds decreases by increasing the number of antennas $M$ for $\delta<M^2$. Therefore, this scheme can provide a deterministic (non random) source of energy when $M\rightarrow\infty$, e.g., $\xi^0_{\mathrm{SA}}=\eta\varrho$.
\end{remark}
\begin{remark}\label{re9}
	For WET strategies without CSI and under i.i.d Rician channels, notice that
	\begin{itemize}
		\item  when $\kappa\rightarrow 0$, e.g., Rayleigh fading, the energy harvested under the $\mathrm{OA}$ and $\mathrm{AA}$ schemes follows  the same distribution, hence no diversity (no gain) is attained from using all antennas at once. For that scenario, the $\mathrm{SA}$ scheme performs best, since although with the same average harvested energy as the others, its variance is the lowest and decreases with $M$;
		\item when $\kappa>0$, the average harvested energy under the $\mathrm{AA}$ scheme is the greatest among all WET schemes, and increases linearly with $M$. However, it is also more dispersed because of its higher variance that also increases linearly with $M$.
	\end{itemize}
\end{remark}
 \begin{figure}[t!]
	\centering
	\subfigure{\includegraphics[width=0.95\columnwidth]{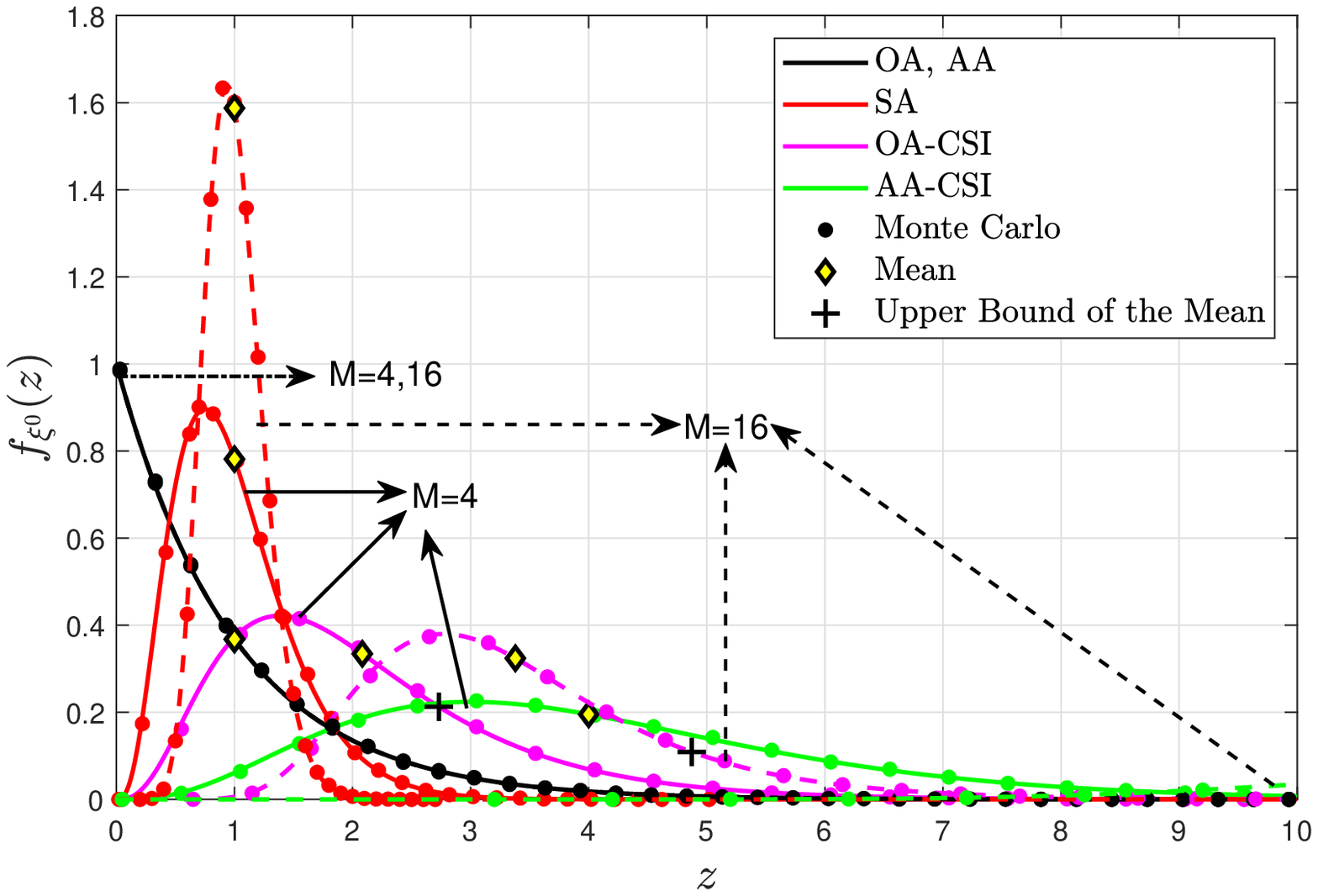}}\\
	\vspace{1mm}
	\subfigure{\includegraphics[width=0.95\columnwidth]{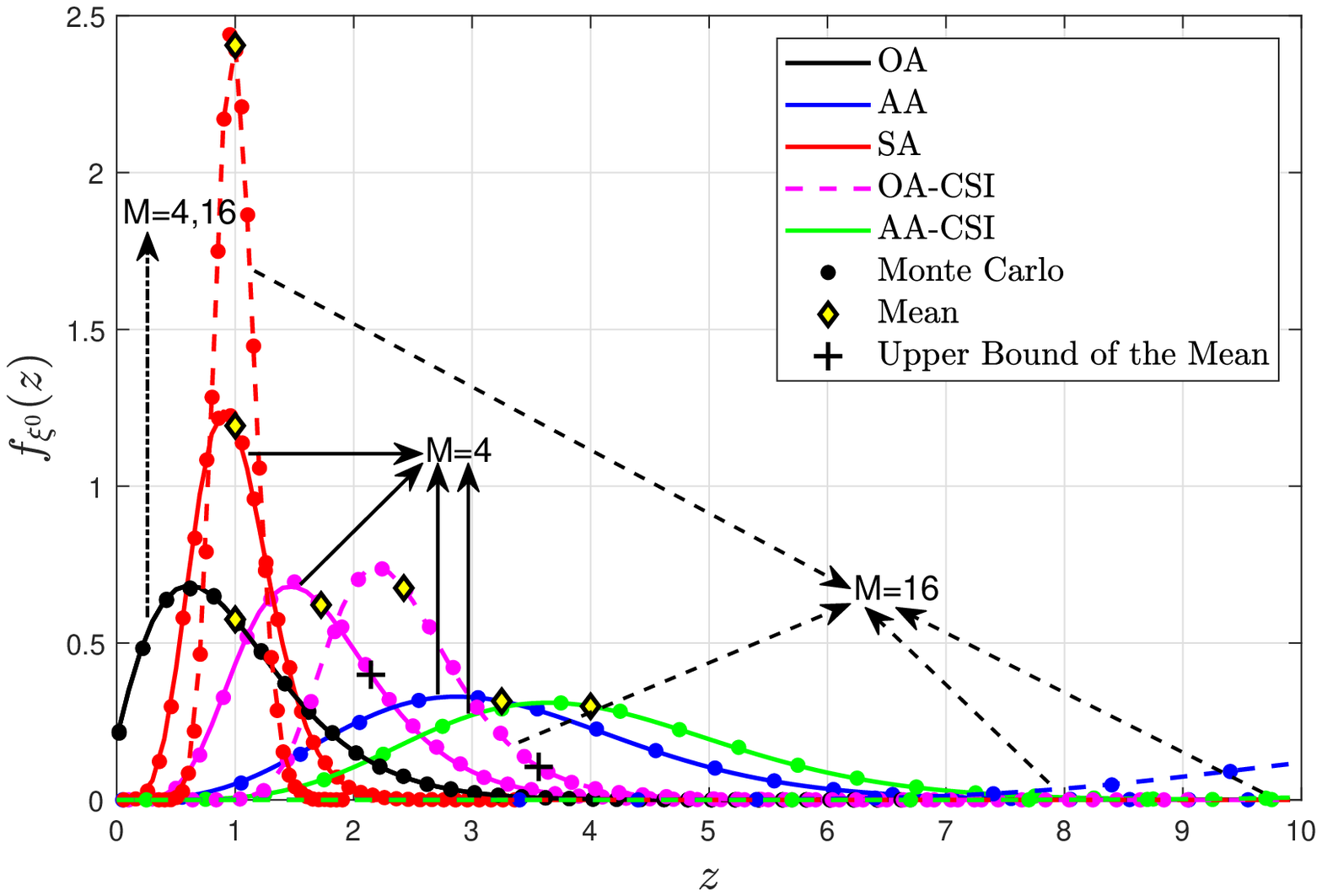}}
	\caption{PDF of $\xi^0$ for $\eta\varrho=1$, $\mathbf{R}=\mathbf{I}$ and $M\in\{4,16\}$, $(a)$ $\kappa=0$ (top), and $(b)$ $\kappa=3$ (bottom).}
	\label{Fig4}
\end{figure}
\begin{figure*}[t!]
	\centering
	\subfigure{\includegraphics[width=0.95\columnwidth]{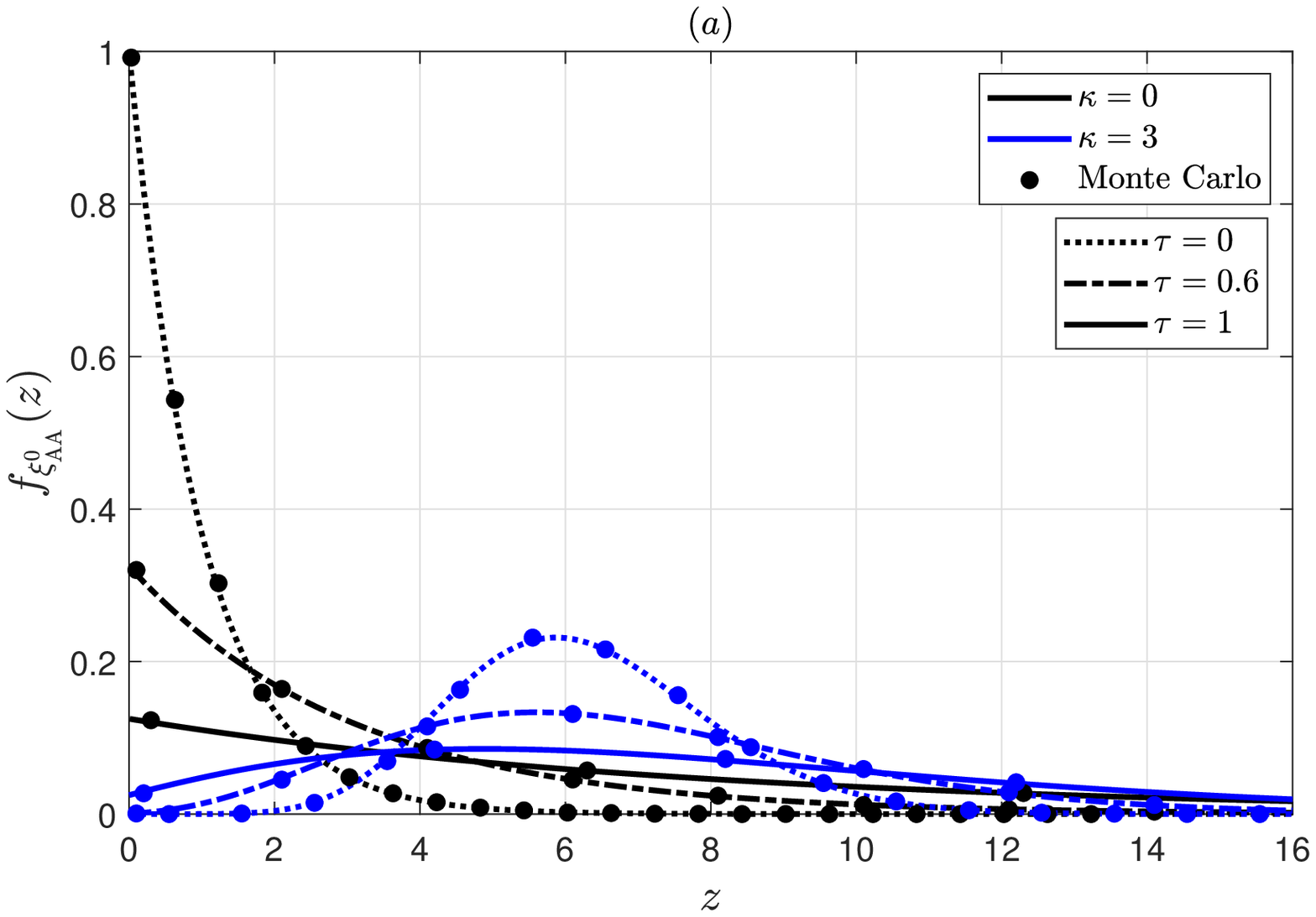}}\ \ \  \subfigure{\includegraphics[width=0.95\columnwidth]{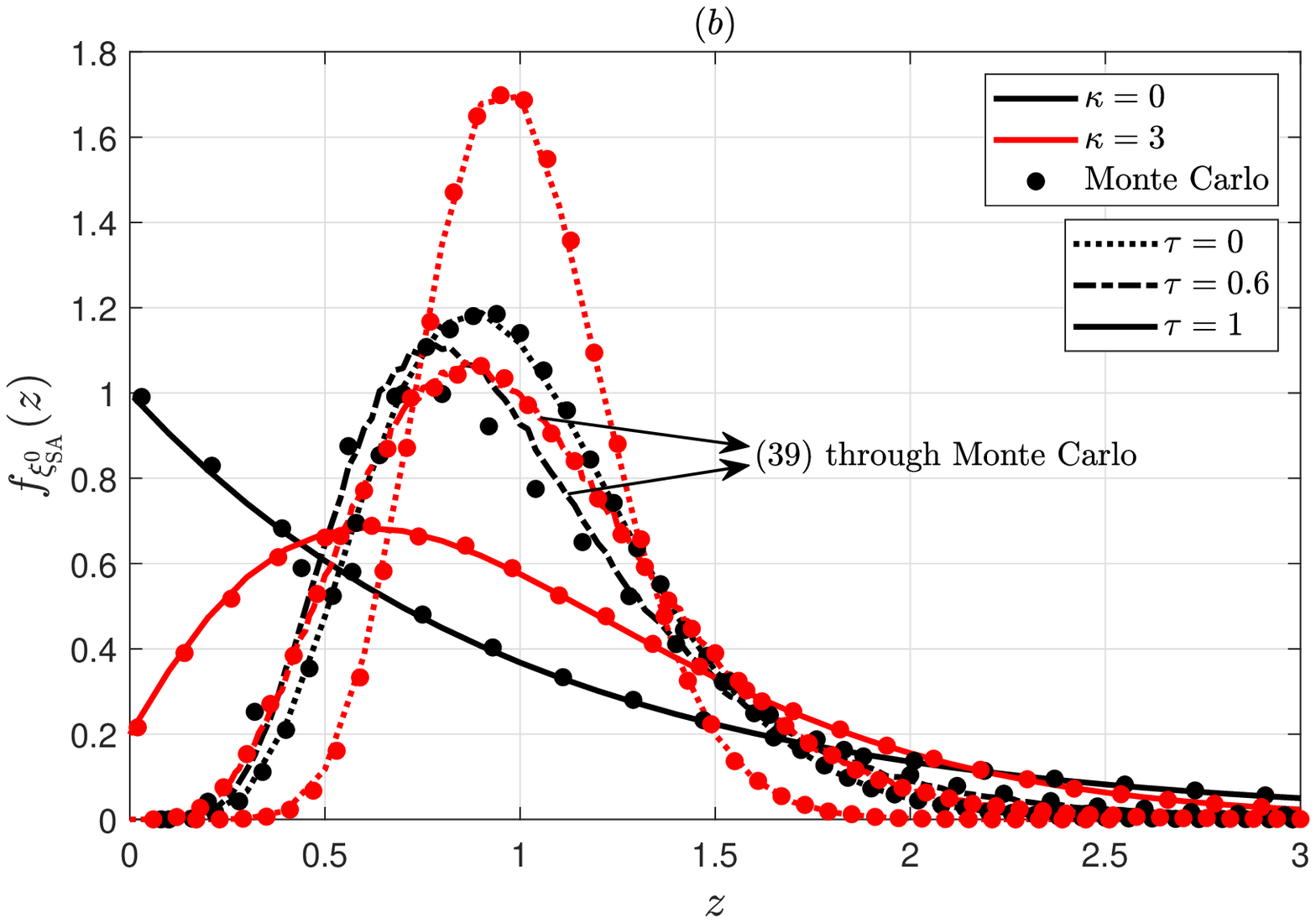}}\\
	\subfigure{\includegraphics[width=0.95\columnwidth]{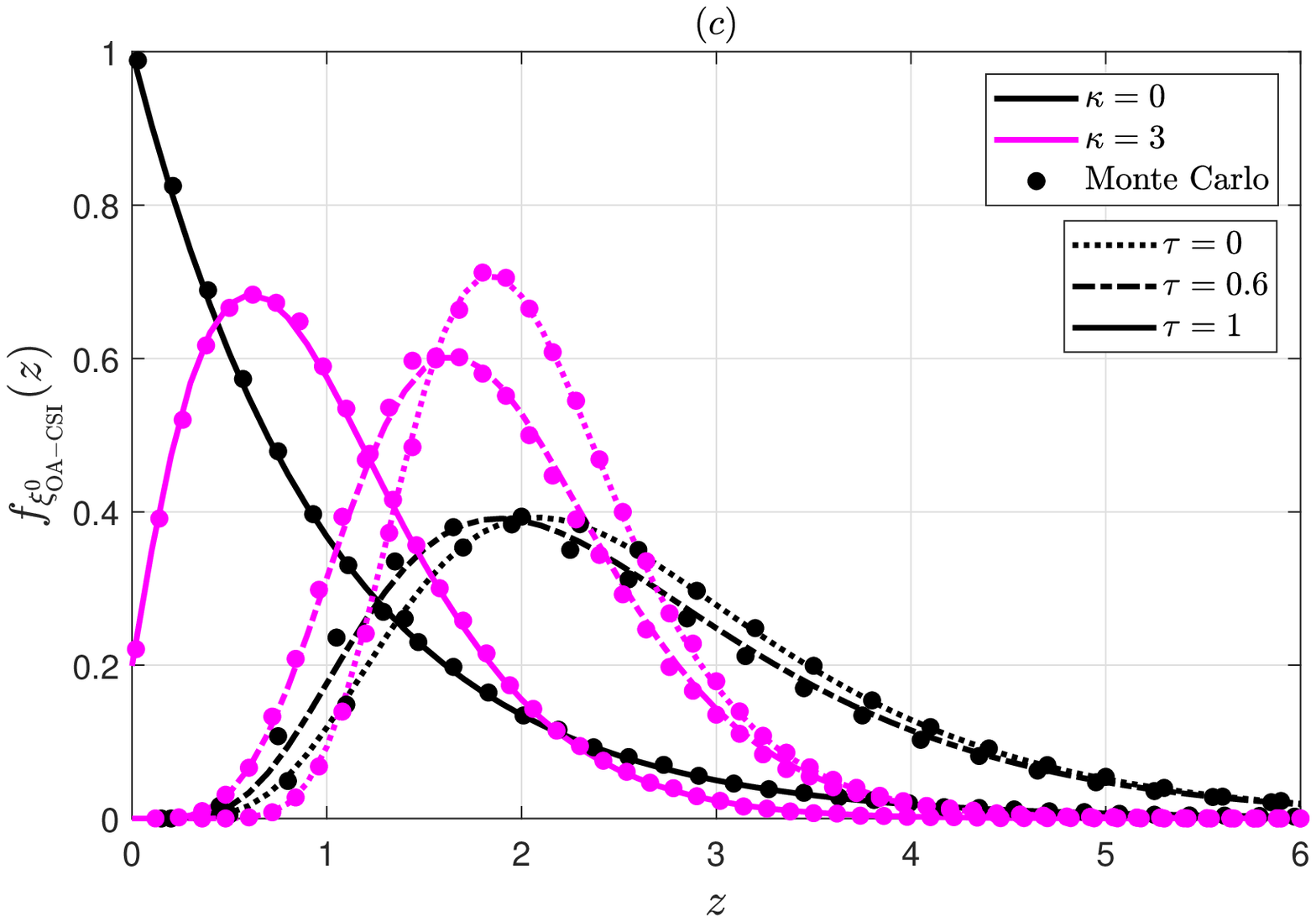}}\ \ \  \subfigure{\includegraphics[width=0.95\columnwidth]{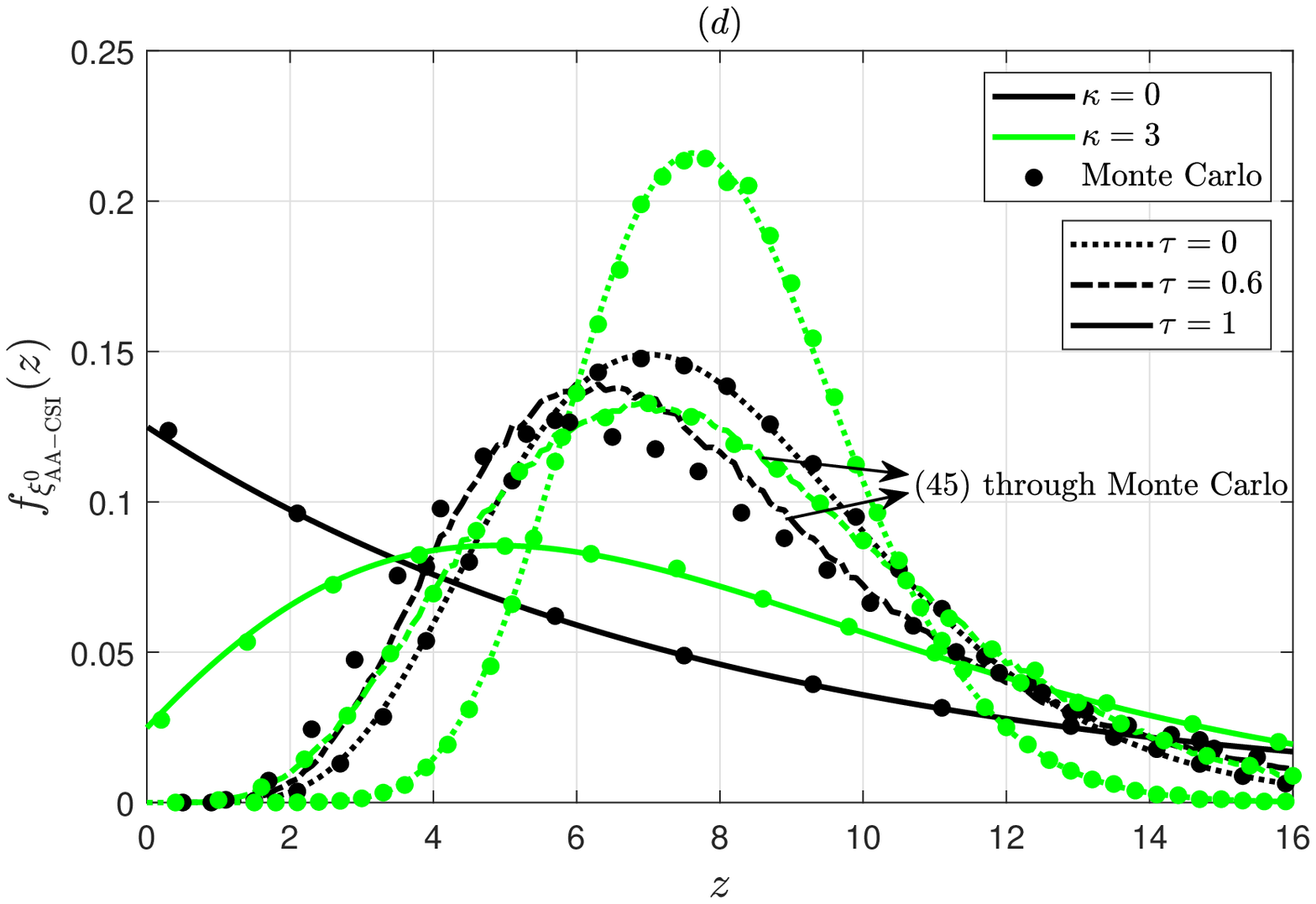}}\\
	\caption{PDF of $\xi^0$ for $\eta\varrho=1$, $M=8$, $\kappa\in\{0,3\}$ and exponential correlation with $\tau\in\{0,0.6,1\}$. $(a)\ \mathrm{AA}$, $(b)\ \mathrm{SA}$, $(c)\ \mathrm{OA-CSI}$ and $(d)\ \mathrm{AA-CSI}$.}
	\label{Fig5}	 	
\end{figure*}
Figs.~\ref{Fig4} and~\ref{Fig5} show the PDF of the harvested energy normalized by $\eta\varrho$, which  allows validating the analytical results by Monte Carlo simulations. Specifically, for Fig.~\ref{Fig4} we consider uncorrelated channels while setting $\kappa\in\{0,3\}$ and $M\in\{4,16\}$, and notice that the upper bound for $\mathbb{E}[\xi^0_{\mathrm{OA-CSI}}]$ provided in \eqref{upp} is also validated, and it is shown to be less tight when $M$ increases.  As expected, the CSI-based schemes outperform their counterparts that do not take advantage of any information. When there is NLOS ($\kappa=0$), it is shown that the $\mathrm{AA}$ and $\mathrm{OA}$ schemes perform the same and without any gain from increasing the number of antennas, thus, in such cases the $\mathrm{SA}$ scheme is the best alternative if no CSI information is available. Notice that by increasing $M$, the PDF of $\xi^0_{\mathrm{SA}}$ gets narrower around its mean, thus, providing a more predictable energy supply. This situation holds for $\kappa>0$, and among the CSI-free schemes the average harvested energy under the $\mathrm{AA}$ scheme is the greatest and increases with $M$, which is a very attractive characteristic.  Meanwhile, in Fig.~\ref{Fig5} the PDF of the harvested energy is shown when using each of the WET schemes\footnote{Since analytical approximations for the distribution of the energy  harvested under the $\mathrm{SA}$ and $\mathrm{AA-CSI}$ schemes were only obtained for $\delta\in\{M,M^2\}\rightarrow\tau\in\{0,1\}$, the expressions were validated through Monte Carlo for the case of $\tau=0.6$.} with $M=8$, $\kappa\in\{0,3\}$ and exponential correlation with $\tau\in\{0,0.6,1\}\rightarrow\delta\in\{8,24.626,64\}$. Notice that correlation impacts negatively on the performance of both CSI-based schemes. Indeed, the greater the correlation the lesser the gains of operating with CSI, and the performance under each CSI-based scheme matches that of its CSI-free counterpart for $\tau=1$. Additionally, comments in Remarks~\ref{re2_5},~\ref{re5} and~\ref{re3} are corroborated, while results also validate the provided analytical  approximations and complement our analysis around Fig.~\ref{Fig2} and \ref{Fig3}.

Fig.~\ref{Fig6}a shows the energy outage performance for a setup with $M=8$ and $\kappa=1$. Notice that for relatively small $\xi_\mathrm{th}$ the schemes can be ordered according to a performance decreasing order as $\mathrm{AA-CSI}>\mathrm{OA-CSI}>\mathrm{SA}>\mathrm{AA}>\mathrm{OA}$. Thus, the $\mathrm{SA}$ scheme performs the best among the CSI-free strategies, while for relatively large $\xi_{\mathrm{th}}$ the $\mathrm{AA}$ scheme is superior, and it could even reach a performance greater than that achieved by $\mathrm{OA-CSI}$. The very idealistic $\mathrm{AA-CSI}$ scheme is the clear winner in this regard also.
A more detailed analysis of the energy outage is presented in Fig.~\ref{Fig6}b, where the regions where each CSI-free WET strategy performs the best are shown for a wide range of values of $\kappa$, exponential correlation parameter $\tau$ and $M\in \{4, 16\}$. The $\mathrm{OA}$ scheme does not appear since it is always outperformed by either $\mathrm{SA}$ or $\mathrm{AA}$ strategy. Notice that $\mathrm{SA}$ benefits less from a greater LOS parameter and/or a greater energy threshold and/or greater spatial correlation. Finally, the number of antennas does not impact strongly on delimiting the regions for which one scheme outperforms the other in terms of the energy outage probability.
\begin{figure}[t!]
	\centering
	\subfigure{\includegraphics[width=0.95\columnwidth]{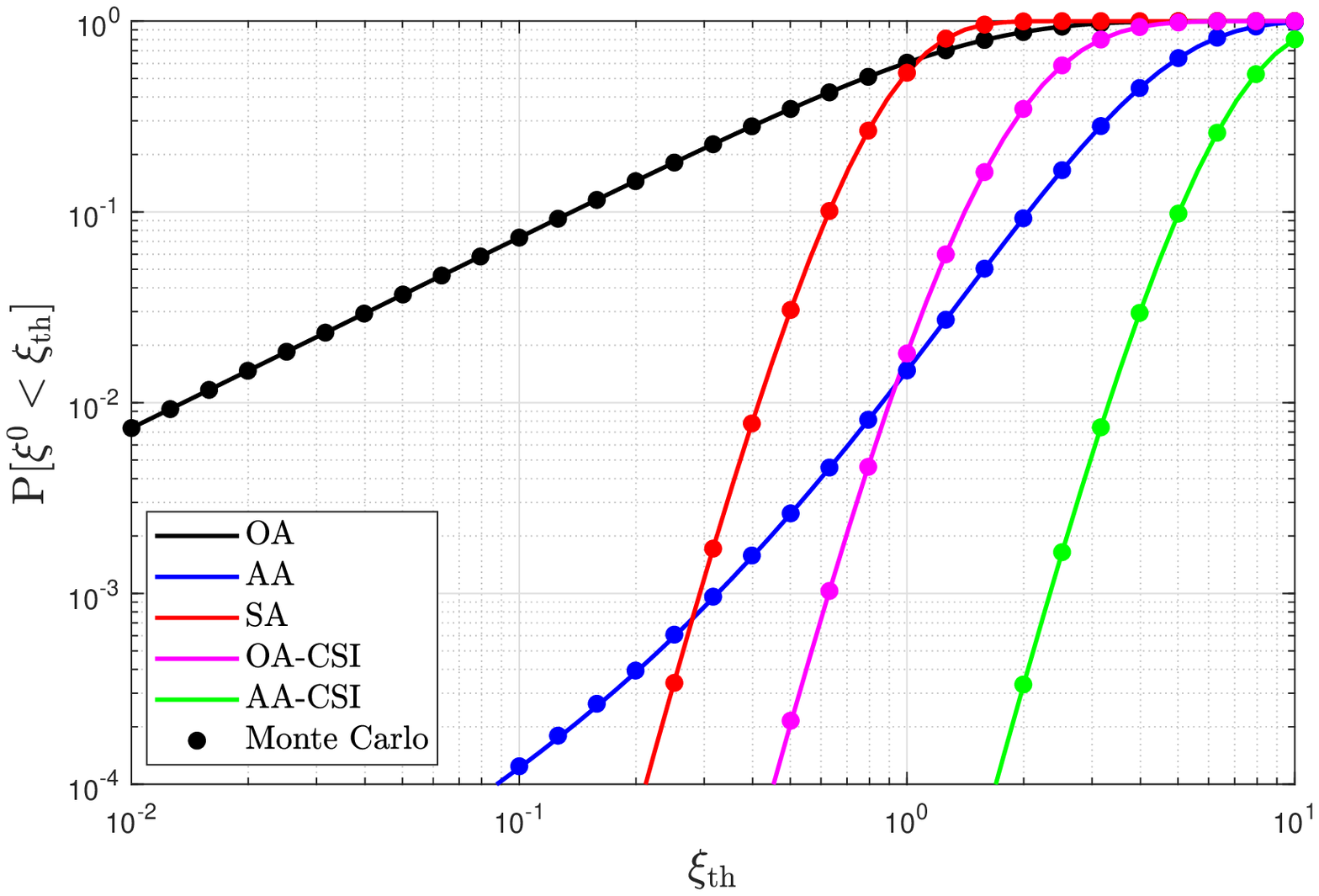}}\\
	\vspace{1mm}
	\subfigure{\includegraphics[width=0.95\columnwidth]{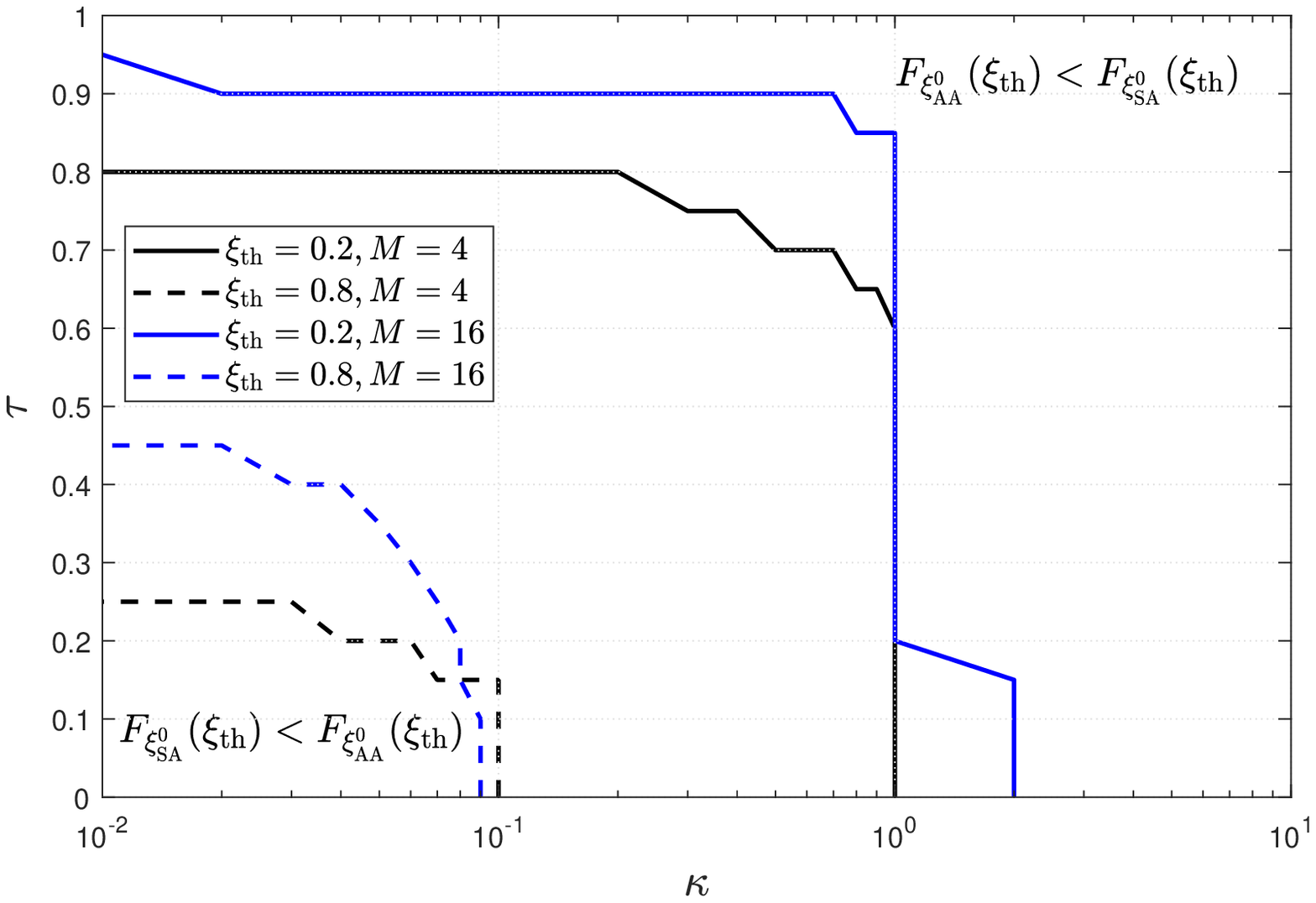}}
	\caption{$(a)$ Probability of energy outage as a function of the minimum amount of energy required for the sensor operation and $\mathbf{R}=\mathbf{I}$, $M=8$ and $\kappa=1$ (top). $(b)$ Region for which $F_{\xi_{\mathrm{SA}}}(z)\!\gtrless\! F_{\xi_{\mathrm{AA}}}(z)$ (bottom). We set $\eta=\varrho=1$.}
	\label{Fig6}
\end{figure}
\subsection{Statistics of the Harvested Energy under Sensitivity and Saturation Phenomena}\label{non}
 The statistics of the harvested energy under the ideal EH linear model, for which the harvested energy is always proportional to the incident RF power with factor $\eta$, was investigated in the previous subsection.  In practice, the sensitivity and saturation phenomena limit also the EH performance. Incident signal powers below the sensitivity level cannot be harvested, while power levels in the saturation region cannot be fully exploited. Therefore, systems with high variance in the RF incident power, e.g., all the CSI-based and the AA scheme, are more sensitive to these phenomena. In this subsection we depart from ideal results in Subsection~\ref{linear} to attain the distribution and main statistics of the harvested energy in more practical setups. 
 \begin{theorem}\label{the2}
 	The energy harvested under each of the schemes when $|\mathcal{S}|=1$ is distributed as
 	\begin{align}\label{cdf}
 	F_{\xi}(\xi)&\cong  \left\{ \begin{array}{ll}
 	F_{\xi^0}(\eta\varpi_1), &  0\le \xi<\eta\varpi_1 \\
 	F_{\xi^0}(\xi), &  \eta\varpi_1\le\xi<\eta\varpi_2 \\
 	1, &  \xi\ge\eta\varpi_2 
 	\end{array}
 	\right.,
 	\end{align}
 	which is only an approximation when using the $\mathrm{SA}$ scheme. The PDF of $\xi$ can be easily derived from \eqref{cdf}.
 \end{theorem}
\begin{proof}
	We have that $\xi^{\mathrm{rf}}=\xi^0/\eta$. Then, for the $\mathrm{OA}$ and $\mathrm{AA}$ schemes, $\xi=g(\xi^{\mathrm{rf}})=g(\xi^0/\eta)$ is satisfied according to \eqref{1a} and \eqref{aao}. Additionally, notice that since $g$ is a non-decreasing function which saturates at some point, $\arg\max_{i=1,...,M}g(\varrho|h_{i}|^2)$ may have several solutions, but definitely one of them is always $\arg\max_{i=1,...,M}\varrho|h_{i}|^2$. Therefore, $\xi_{\mathrm{OA-CSI}}=g(\xi^0_{\mathrm{OA-CSI}}/\eta)$ is also satisfied. Similarly, for the $\mathrm{AA-CSI}$ scheme we have that $\arg\max_{\mathbf{w}_1}|\mathbf{h}^T\mathbf{w}_1|^2$ is at least one of the optimum beamformers that maximizes $g\big(\varrho|\mathbf{h}^T\mathbf{w}_1|^2\big)$, thus, $\xi_{\mathrm{AA-CSI}}=g(\xi^0_{\mathrm{AA-CSI}}/\eta)$. For the $\mathrm{SA}$ scheme the situation is different since in general $\frac{1}{M}\sum_{i=1}^{M}g\big(\varrho|h_{i}|^2\big)\ne g\big(\frac{\varrho}{M}\sum_{i=1}^{M}|h_{i}|^2\big)$, however it could be used as an approximation that becomes tighter as $\rho$ increases. Therefore,
	 \begin{align}\label{app}
	 \xi \cong g(\xi^0/\eta),
	 \end{align}
	 which is satisfied with equality for all the schemes but for $\mathrm{SA}$. Using \eqref{app} and \eqref{conv} it is straightforward attaining \eqref{cdf}.
\end{proof}

Notice that for $\varpi_1=0$, $\varpi_2\rightarrow\infty$ and each WET scheme, \eqref{cdf} matches exactly the CDF of the harvested energy given in the previous subsection. On the other hand, previously we have also attained the average and variance of the harvested energy, and consequently for the RF incident power too.  As commented, the variance analysis allows explaining the energy outage performance in conjunction with the CDF curves. Henceforth, we focus only on the average harvested energy when considering the sensitivity and saturation phenomena. The following lemma provides an analytical result that is useful for the average harvested energy characterization given right after in Theorem~\ref{the3}.
\begin{lemma}\label{lem1}
	Let $q_1(\varphi,\phi)=\mathbb{E}[g(\phi Z_1)]$ and $q_2(\psi,\phi)=\mathbb{E}[g(\phi Z_2)]$, where $Z_1\sim\chi^2(\varphi,0)$ and $Z_2\sim\chi^2(2,\psi)$, then
	\begin{align}
	q_1(\varphi,\phi)&=\frac{2\eta}{\Gamma(\varphi/2)}\bigg(\phi\Gamma\Big(1+\frac{\varphi}{2},\frac{\varpi_1}{2\phi}\Big)-\phi\Gamma\Big(1+\frac{\varphi}{2},\frac{\varpi_2}{2\phi}\Big)+\nonumber\\
	&\qquad\qquad\qquad\qquad\qquad\ \ +\frac{\varpi_2}{2}\Gamma\Big(\frac{\varphi}{2},\frac{\varpi_2}{2\phi}\Big)\bigg),\label{g1}\\
	q_2(\psi,\phi)&\approx \frac{\eta}{\Gamma(m)}\bigg[\frac{\phi(\psi+2)}{m}\bigg(\Gamma\Big(m+1,\frac{m\varpi_1}{\phi(\psi+2)}\Big)+\nonumber\\
	&\!-\!\Gamma\Big(m\!+\!1,\!\frac{m\varpi_2}{\phi(\psi\!+\!2)}\Big)\!\bigg)\!\!+\!\varpi_2\Gamma\Big(m,\!\frac{m\varpi_2}{\phi(\psi\!+\!2)}\Big)\bigg],\label{g2}
	\end{align}
	where $g$ is the function defined in \eqref{conv} and $m=(\psi/2+1)^2/(\psi+1)$.
\end{lemma}
\begin{proof}
	According to \eqref{fZ} and \eqref{FZ}, the distribution of $Z_1$ can be simplified as follows
	\begin{align}
	f_{Z_1}(z)&=\frac{1}{2}z^{\varphi/4-1/2}e^{-z/2}\lim\limits_{\psi\rightarrow 0}\frac{I_{\varphi/2-1}(\sqrt{\psi z})}{\psi^{\varphi/4-1/2}}\nonumber\\
	&\stackrel{(a)}{=}\frac{1}{2}z^{\varphi/4-1/2}e^{-z/2} \frac{1}{\Gamma(\varphi/2)}\Big(\frac{z}{4}\Big)^{\varphi/4-1/2}\nonumber\\
	&\stackrel{(b)}{=}\frac{1}{2\Gamma(\varphi/2)}\Big(\frac{z}{2}\Big)^{\varphi/2-1}e^{-z/2}\label{new1} \\
	F_{Z_1}(z)&=1-\mathcal{Q}_{\varphi/2}(0,\sqrt{z})=1-\frac{\Gamma(\varphi/2,z/2)}{\Gamma(\varphi/2)},\label{new2}
	\end{align}
	where $(a)$ comes from finding the limit by using l'H\^opital  rule successively and $(b)$ follows after simple algebraic transformations. For getting \eqref{new2}, it is required evaluating $\mathcal{Q}(0,\sqrt{z})$ by using its definition and performing some limit operations or equivalently one may use the result in \eqref{new1} for finding the CDF; in any case the result matches \eqref{new2}. Now, using \eqref{new1} and \eqref{new2} along with \eqref{conv} to calculate $\mathbb{E}[g(\phi Z_1)]$ we obtain
	\begin{align}
	\mathbb{E}[g(\phi Z_1)]&=\eta\phi\int_{\varpi_1/\phi}^{\varpi_2/\phi}zf_{Z_1}(z)\mathrm{d}z+\eta\varpi_2\int_{\varpi_2/\phi}^{\infty}f_{Z_1}(z)\nonumber\\
	&=\!\frac{\eta}{\Gamma(\varphi/2)}\!\!\int_{\frac{\varpi_1}{\phi}}^{\frac{\varpi_2}{\phi}}\!\!\Big(\frac{z}{2}\Big)^{\frac{\varphi}{2}}\!\!e^{-\frac{z}{2}}\mathrm{d}z\!+\!\eta\varpi_2\Big(1\!-\!F_{Z_1}\big(\frac{\varpi_2}{\phi}\big)\Big)\nonumber\\
	&=\frac{\eta\phi}{\Gamma(\varphi/2)}\big(-2\Gamma(1+\varphi/2,z/2)\big)\Big|_{\varpi_1/\phi}^{\varpi_2/\phi}+\nonumber\\
	&\qquad\qquad\qquad\ \ \ +\eta\varpi_2\frac{\Gamma(\varphi/2,\varpi_2/(2\phi))}{\Gamma(\varphi/2)},
	\end{align}
	and \eqref{g1} follows after some simple algebraic transformations. Continuing with the second part of the proof, the PDF of $Z_2$ can be simplified according to  \eqref{fZ} as follows
	\begin{align}
	f_{Z_2}(z)&=\frac{1}{2}e^{-(z+\psi)/2}I_0(\sqrt{\psi z}),
	\end{align}
	which is the power distribution under Rician fading channel with $\kappa=\psi/2$ and average power $(\psi+2)$ \cite{Proakis.2001}. Unfortunately, to the best of our knowledge computing $\int\!\! zf_{Z_2}(z)\mathrm{d}z$ is not analytically tractable, thus, we were not able of finding an exact closed form expression. However we can take advantage of the fact that for  $m=(\kappa+1)^2/(2\kappa+1)$, the Rician distribution approximates the Nakagami-m \cite{Goldsmith.2005}, then $Z_2\sim\Gamma\Big(\frac{(\psi/2+1)^2}{\psi+1},\frac{(\psi+2)(\psi+1)}{(\psi/2+1)^2}\Big)$ approximately. Now, 
	\begin{align}
	\mathbb{E}[g(\phi Z_2)]&=\eta\phi\int_{\varpi_1/\phi}^{\varpi_2/\phi}zf_{Z_2}(z)\mathrm{d}z+\eta\varpi_2\int_{\varpi_2/\phi}^{\infty}f_{Z_2}(z)\mathrm{d}z\nonumber\\
	&\!\stackrel{(a)}{\approx} \!\eta\phi\!\!\!\!\!\!\int\limits_{\varpi_1/\phi}^{\varpi_2/\phi}\!\!\!\!\!\!\Big(\!\frac{m}{\psi\!+\!2}\!\Big)^m\!\frac{z^me^{-\frac{mz}{\psi\!+\!2}}}{\Gamma(m)}\mathrm{d}z\!+\!\eta\varpi_2\Big(\!1\!-\!F_{Z_2}\!\big(\!\frac{\varpi_2}{\phi}\!\big)\!\Big)\nonumber\\
	&\stackrel{(b)}{=} -\frac{\eta\phi(\psi+2)}{m\Gamma(m)}\Gamma\Big(m+1,\frac{mz}{\psi+2}\Big)\Big|_{\varpi_1/\phi}^{\varpi_2/\phi}+\nonumber\\
	&\qquad\qquad\qquad\qquad\ \ +\eta\varpi_2\frac{\Gamma\big(m,\frac{m\varpi_2}{\phi(\psi+2)}\big)}{\Gamma(m)},
	\end{align}
	where $(a)$ comes from using \eqref{fV} and $(b)$ follows from solving the integral and using \eqref{FV}. Then, after some simple algebraic transformations we attain \eqref{g2}.
\end{proof}
\begin{theorem}\label{the3}
	The average harvested energy under sensitivity and saturation phenomena, and for each of the WET schemes (but for $\mathrm{OA-CSI}$) when $|\mathcal{S}|=1$, is given by
	\begin{align}
	\mathbb{E}[\xi_{\mathrm{OA}}]&=q_2\Big(2\kappa,\frac{\varrho}{2(1+\kappa)}\Big),\label{1an} \\
	\mathbb{E}[\xi_{\mathrm{AA}}]&=q_2\Big(\frac{2M^2\kappa}{\delta},\frac{\delta\varrho}{2M(1+\kappa)}\Big), \label{aaon}\\
	\mathbb{E}[\xi_{\mathrm{SA}}]&\approx q_1\Big(2(M-1),\frac{(M^2-\delta)\varrho}{2M^2(M-1)(1+\kappa)}\Big)+\nonumber\\
	&\qquad\qquad\ \ \ +q_2\Big(\frac{2\kappa M^2}{\delta},\frac{\delta\varrho}{2M^2(1+\kappa)}\Big), \label{oatn}\\
	\mathbb{E}[\xi_{\mathrm{AA-CSI}}]&=q_1\Big(2(M-1),\frac{(M^2-\delta)\varrho}{2M(M-1)(1+\kappa)}\Big)+\nonumber\\
	&\qquad\qquad\ \ \ +q_2\Big(\frac{2\kappa M^2}{\delta},\frac{\delta\varrho}{2M(1+\kappa)}\Big). \label{aacsin}
	\end{align}
\end{theorem}
\begin{proof}
	According to \eqref{1ar} we have that $\xi_{\mathrm{OA}}^0/\eta\!\sim\! \frac{\varrho}{2(1+\kappa)}\chi^2(2,2\kappa)$, then, using \eqref{app} yields $\mathbb{E}[\xi_{\mathrm{OA}}]=\mathbb{E}[g(\xi_{\mathrm{OA}}^0/\eta)]$, and using the results in Lemma~\ref{lem1} we attain \eqref{1an}. Similarly, using \eqref{aaor}, \eqref{pdfc_app} and \eqref{aaacsir} along with \eqref{app}, we have that $\xi_{\mathrm{AA}}=\xi_{\mathrm{AA}}^0/\eta\sim\frac{\varrho\delta}{2M(1\!+\!\kappa)}\chi^2\Big(2,\frac{2M^2\kappa}{\delta}\Big)$, $\xi_{\mathrm{SA}}\approx \xi_{\mathrm{SA}}^0/\eta\sim \frac{\varrho}{2M^2(1\!+\!\kappa)}\Big[\frac{M^2\!-\!\delta}{M\!-\!1}\chi^2\Big(2(M\!-\!1),0\Big) \!+\!\delta\chi^2\Big(2,\frac{2\kappa M^2}{\delta}\Big)\Big]$ and $\xi_{\mathrm{AA-CSI}}=M \xi_{\mathrm{SA}}^0/\eta$, respectively, and after using Lemma~\ref{lem1} we attain \eqref{aaon}, \eqref{oatn} and \eqref{aacsin}.
\end{proof}
Notice that a closed-form expression for the average harvested energy under $\mathrm{OA-CSI}$ was not tractable using the ideal linear EH model, thus, neither here. Finally, since $q_2$ is given approximately in \eqref{g2}, the analytical characterizations of the average harvested energy provided in Theorem~\ref{the3} are approximations given they all depend on $q_2$; hence, they approximate also their counterparts in the previous subsection when $\varpi_1=0$ and $\varpi_2\rightarrow\infty$.
\section{Numerical Results and Discussion}\label{results}
In this section we investigate the performance of the WET strategies under Rician fading and exponential correlation. Based on the experimental data in \cite{Le.2008}, we assume  sensors equipped with EH hardware characteristics $\varpi_1=-22$ dBm, $\varpi_2=-4.8$ dBm and $\eta=25\%$\footnote{Rather than the maximum efficiency, we assume a more conservative value that allows fitting better the entire experimental data. Notice also that operating with such conversion efficiency is viable for many other EH devices as reported in \cite{Valenta.2014}.}. In Fig.~\ref{Fig7} we show how our model fits the experimental data; while we also consider the non-linear harvesting model proposed in \cite{Boshkovska.2015} which is based on a logistic function. Specifically, the non-linear model is given by $\xi=p_3\Big(\frac{1+e^{p_1p_2}}{1+e^{-p_1(\xi^{\mathrm{rf}}-p_2)}}-1\Big)/e^{p_1p_2}$, where $p_1,\ p_2,\ p_3$ can be easily found by a standard curve fitting tool. Notice that the non-linear model fits more accurately the data in almost the entire region of input powers than our EH linear model; however, it fails in modeling the sensitivity phenomenon since if, and only if, $\xi^\mathrm{rf}=0\rightarrow \xi=0$. 

Next, we consider a setup with $|\mathcal{S}|=1$ and show in Fig.~\ref{Fig8} the average harvested energy as a function of $\varrho$  under the non-linear EH model \cite{Boshkovska.2015}, and the ideal and non ideal EH linear models considered in this work;
while for short, we do not present the performance results of the $\mathrm{OA}$ scheme. 
\begin{figure}[t!]
	\centering
	\subfigure{\includegraphics[width=0.95\columnwidth]{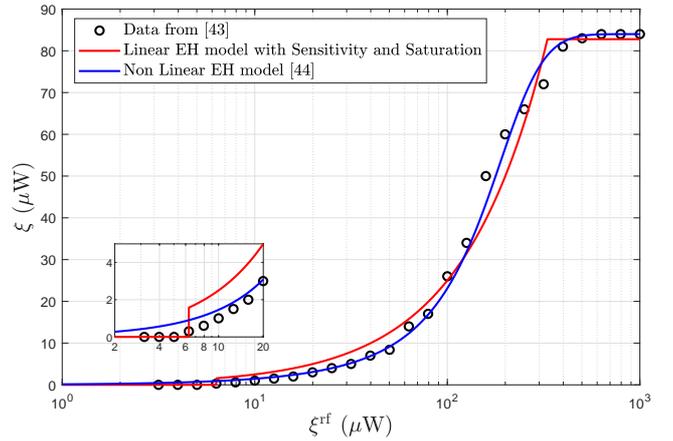}}
	\caption{A comparison between the harvested power for (\textit{i}) the utilized model taking into account sensitivity and saturation phenomena, (\textit{ii}) the non-linear EH model in \cite{Boshkovska.2015}, and (\textit{iii}) the measurement data from a practical EH circuit
\cite{Le.2008}. Configuration parameters are: $\varpi_1=-22$ dBm, $\varpi_2=-4.8$ dBm, $\eta=25\%$, while through curve fitting we attain: $p_1=0.015$, $p_2=140$ and $p_3=84$.}
	\label{Fig7}	 	
\end{figure}
\begin{figure*}[t!]
	\centering
	\subfigure{\includegraphics[width=0.95\columnwidth]{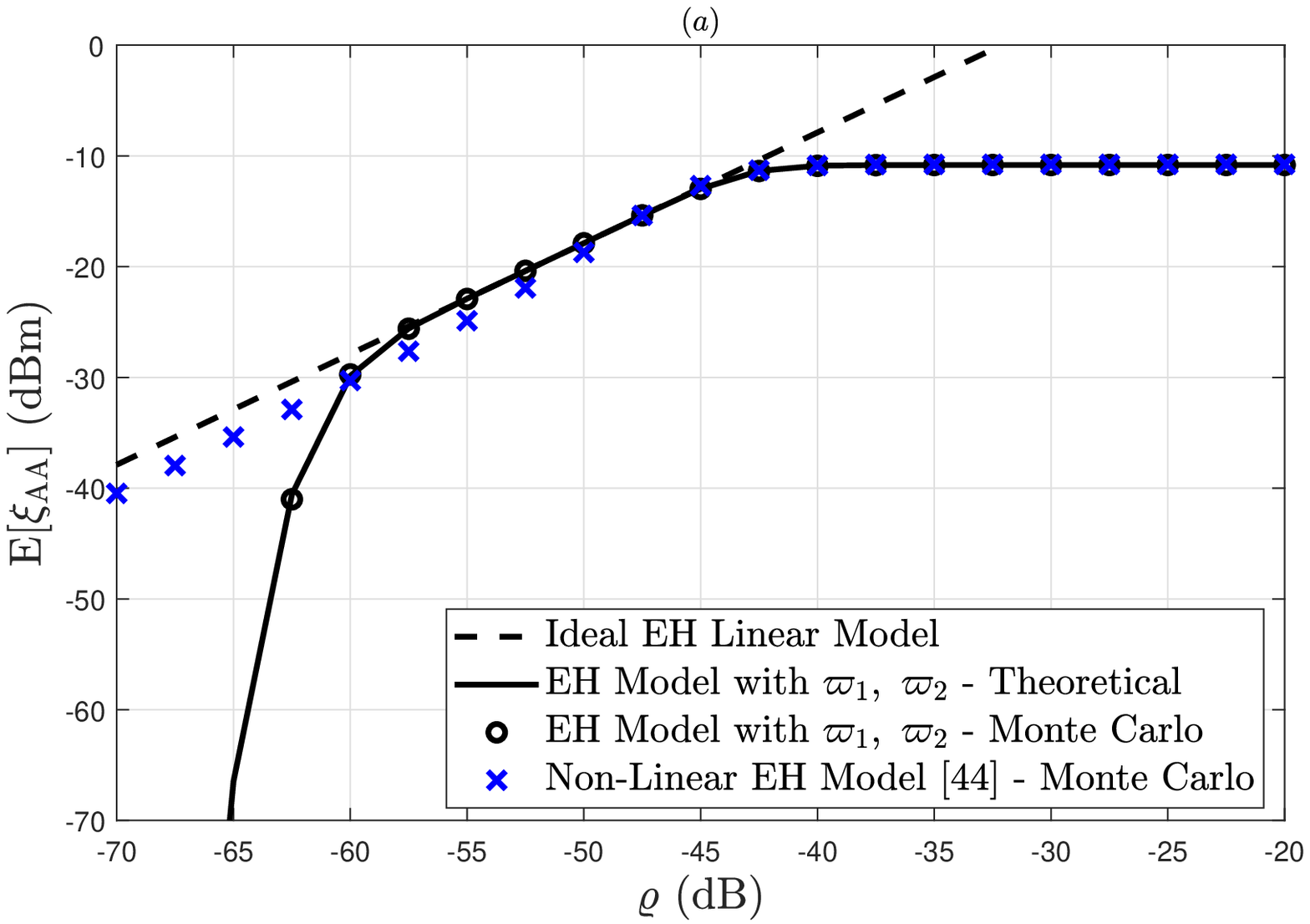}}\ \ \ \ \   \subfigure{\includegraphics[width=0.95\columnwidth]{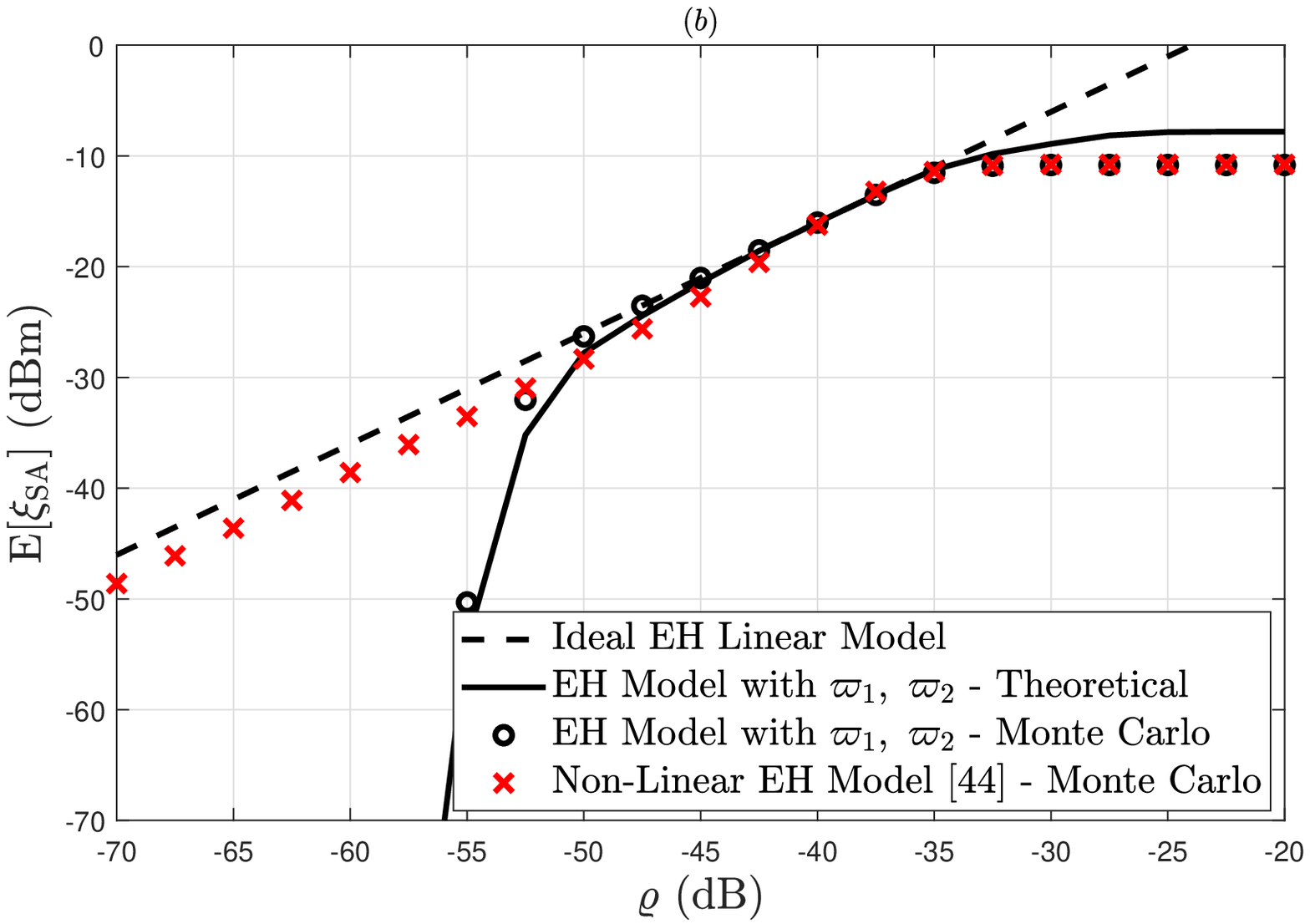}}\\
	\subfigure{\includegraphics[width=0.95\columnwidth]{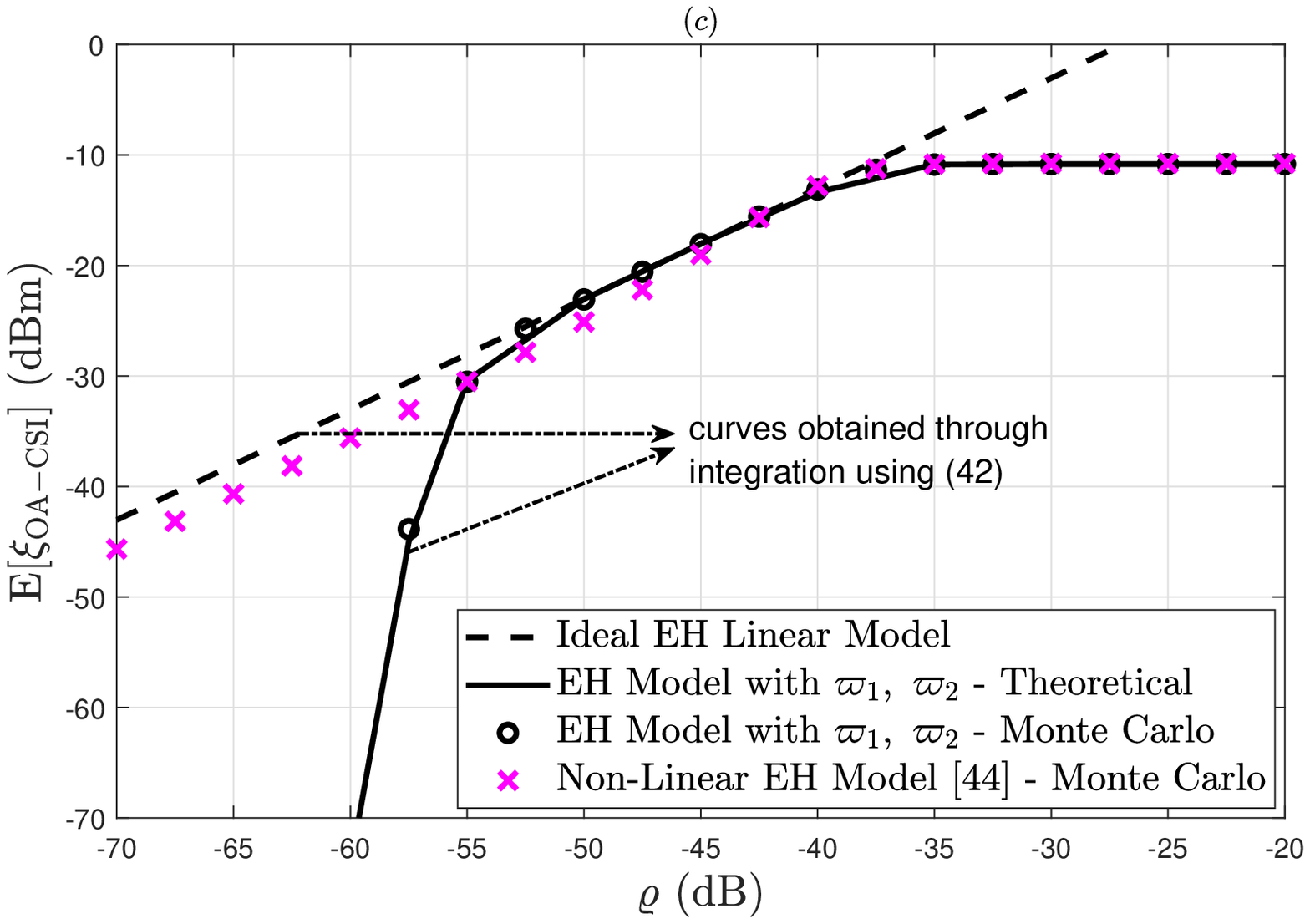}}\ \ \ \ \  \subfigure{\includegraphics[width=0.95\columnwidth]{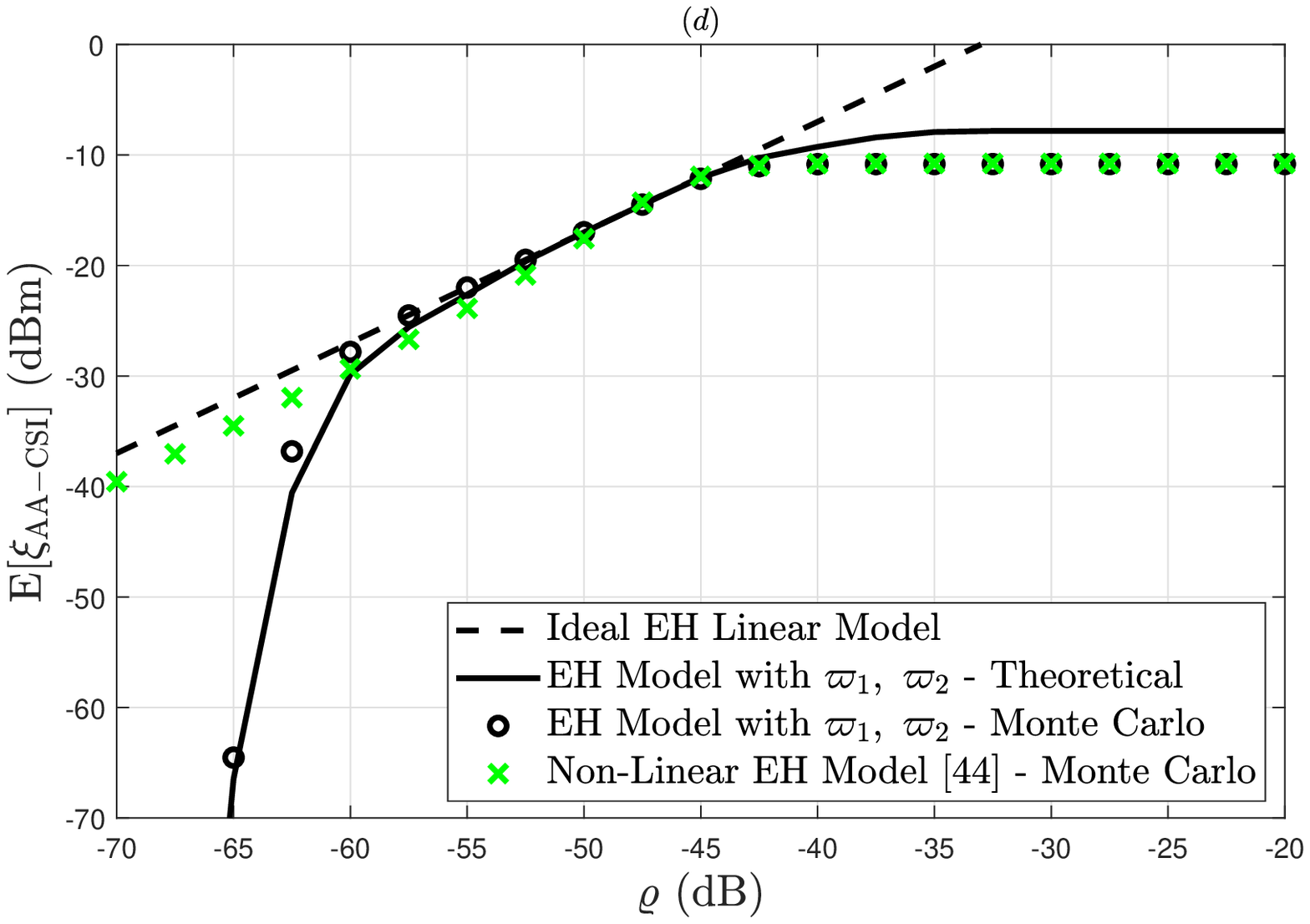}}\\
	\caption{Average harvested energy as a function of $\varrho$ under ideal and non ideal EH linear model, and non-linear model of \cite{Boshkovska.2015}, for $M=8$, $\kappa=3$ and $\tau=0.4$. $(a)\ \mathrm{AA}$, $(b)\ \mathrm{SA}$, $(c)\ \mathrm{OA-CSI}$ and $(d)\ \mathrm{AA-CSI}$.}
	\label{Fig8}	 	
\end{figure*}
Obviously, the average harvested energy  performance is an increasing function of $\varrho$, which is related to the average RF signal power. Notice that for the ideal EH linear model the average harvested energy increases linearly, but when operating under the non-ideal linear and non-linear EH models that is not longer the case. The ideal EH model overestimates heavily the average harvested energy when $\varrho$ is small or large, thus, it is only suitable for a finite range of $\varrho$. This range depends on the system setup and on the utilized WET scheme. For instance, the performance of the three EH models are similar when $\varrho\in[-50,-35]$ dB for the $\mathrm{SA}$ and $\mathrm{OA-CSI}$ schemes, while the range is approximately shifted to the left $10\log_{10}(M)$ dB when utilizing the $\mathrm{AA}$ and $\mathrm{AA-CSI}$ schemes. This is because under the ideal EH linear model, $\mathrm{AA}$ and $\mathrm{AA-CSI}$ schemes share similar average harvested energy performance, and approximately $M$ times greater than the performance attained under $\mathrm{SA}$ and $\mathrm{OA-CSI}$ schemes. Notice that our non-ideal linear EH model, which explicitly considers the sensitivity and saturation phenomena, succeeds in matching the performance of the non-linear EH model proposed in \cite{Boshkovska.2015} when operating under medium to high average input powers for which this model was proven to be accurate. In fact, and because of the sensitivity phenomenon, our non-ideal linear model predicts an abrupt decay in the average harvested energy when $\varrho$ decreases to very small values; what does not happen under the non-linear model that does not take strictly into account this phenomenon, hence, for small $\varrho$ the accuracy of our model is foreseen to be tight to practice. Additionally, analytical expressions provided in Theorem~\ref{the3} are corroborated also in Fig.~\ref{Fig8}. The gap between simulation and analytical results come from the fact that function $q_2$ in Theorem~\ref{the3} is approximated according to \eqref{g2}, however, in all the cases analytical results capture well the impact of sensitivity and saturation impairments.

Given the validation of our non-ideal  EH linear model, we use it next for analyzing the system performance when powering multiple sensors, e.g., $|\mathcal{S}|\ge 1$. To that end, we consider that sensors  $\mathcal{S}$ are uniformly distributed in the disk region of radius $R=10$m \cite{Clerckx.2018,Bi.2016} around $T$, thus, the distance between $S_j$ and $T$, denoted as $\iota_j$, is distributed with $f_{\iota_j}(x)=2x/R^2,\ x\le R$. We assume that $\varrho_j=\iota_j^{-3}/50$, which may model practical scenarios with path-loss exponent $3$, transmit power of $40$ dBm, and $27$ dB of average signal power attenuation at a reference distance of $1$m \cite{Goldsmith.2005}. Also, the LOS factors and correlation coefficients are expected to decrease with the distance in practical setups, hence we take this into consideration and assume an exponential decay given by $\kappa_j=10\times e^{-\iota_j/2},\ \tau_j=e^{-\iota_j/3}$. For example, this means that at $1$m$\rightarrow \kappa_j\approx 6.07,\ \tau\approx 0.72$, while at the disk edge these values drop to $\kappa_j\approx 0.07,\ \tau\approx 0.04$. Notice that, under the CSI-free schemes, the performance metrics such as average harvested energy per user and average energy outage per user, can be easily computed by departing from results in Section~\ref{Ric} and performing integration over $f_{\iota_1}(x)$. For the CSI-based schemes, which are used only as benchmarks, mathematical analyses in multi-user setups and under sensitivity and saturation phenomena are cumbersome and we resort to Monte Carlo methods and global optimization solvers when finding optimum energy beamformings. Notice that this may be impossible to implement in practice and the performance gains should decrease with the cardinality of $\mathcal{S}$. These are some of the main reasons for using the CSI-free WET strategies and the scope of this work.

\begin{figure}[t!]
	\centering
	\subfigure{\includegraphics[width=0.95\columnwidth]{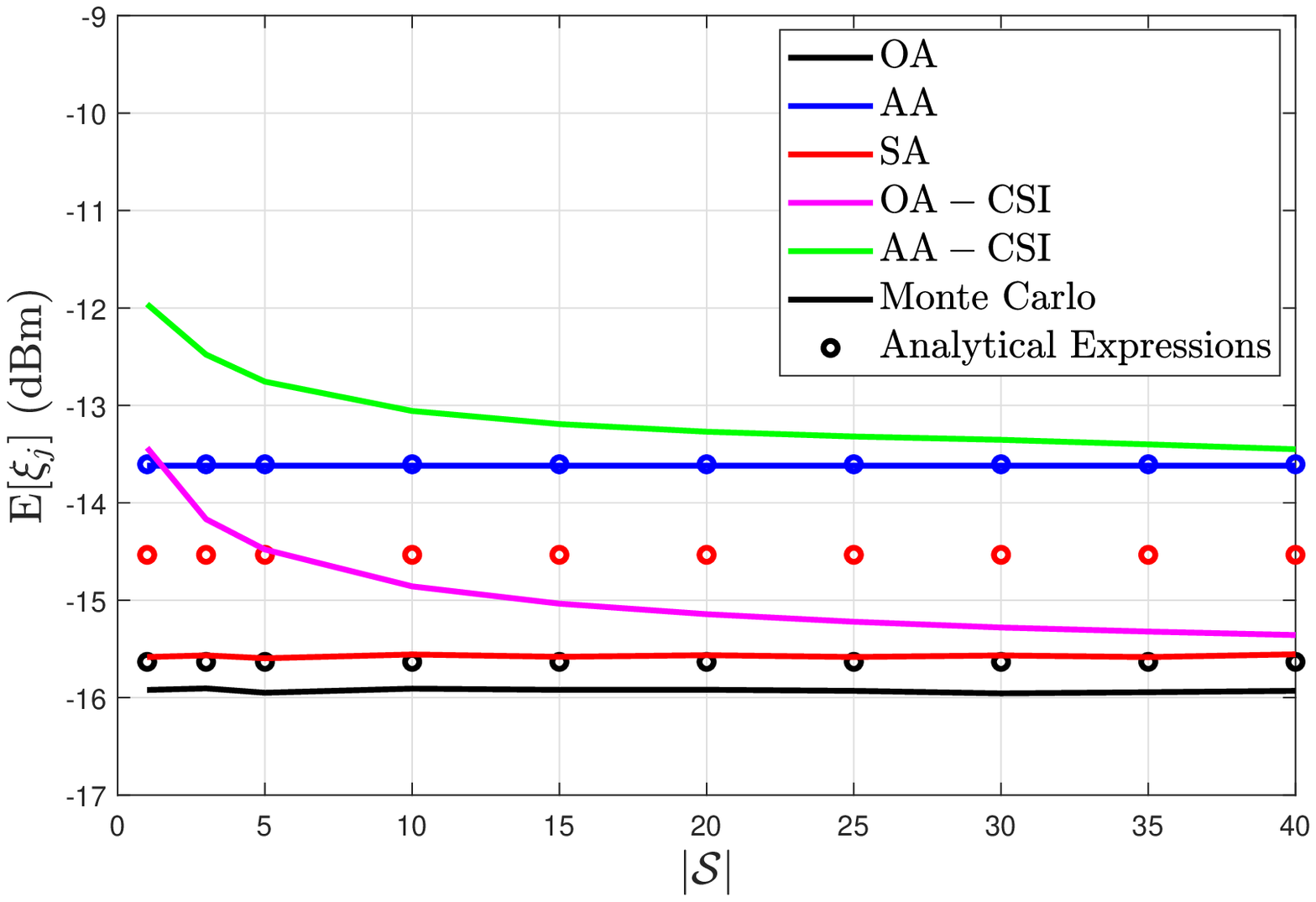}}\\
	\subfigure{\includegraphics[width=0.95\columnwidth]{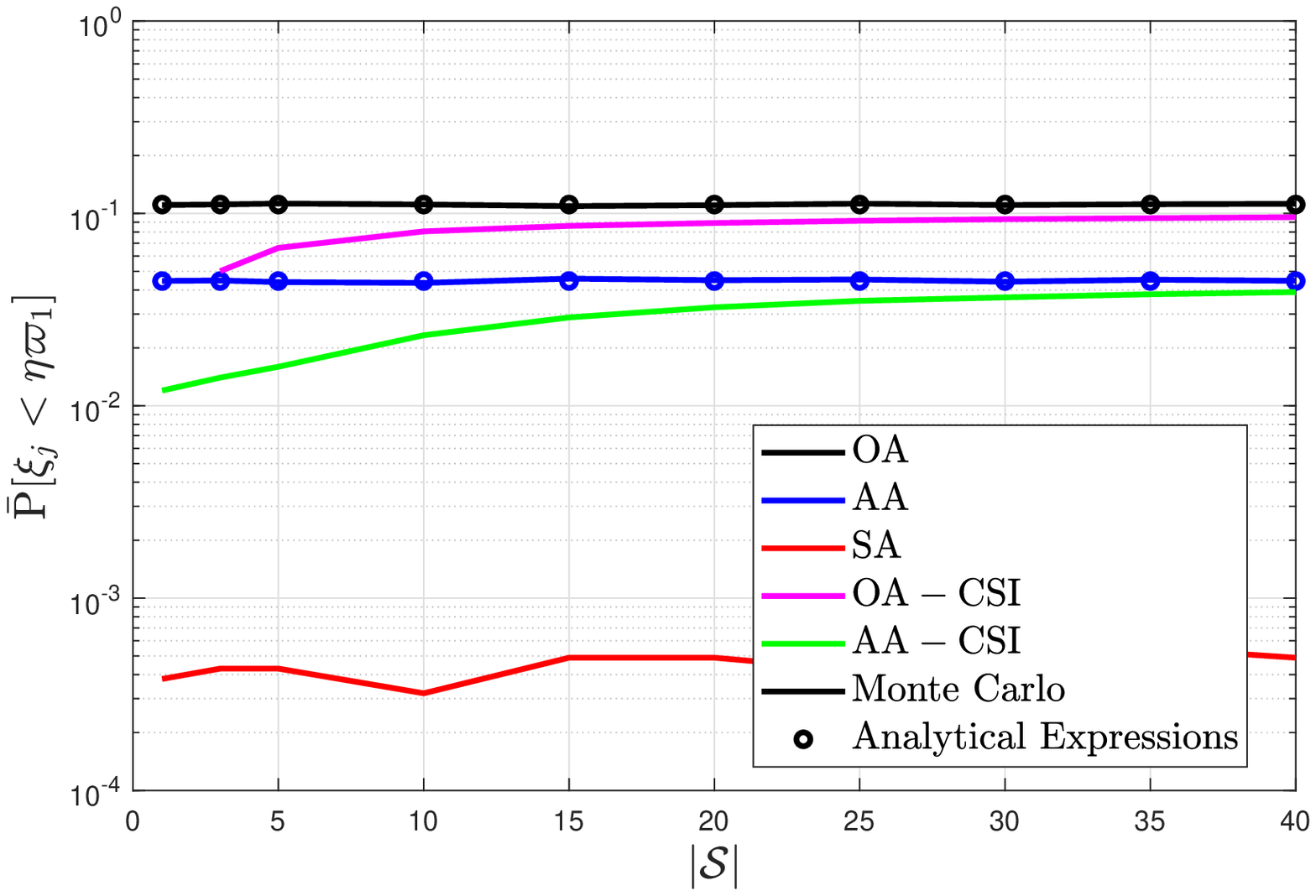}}
	\caption{$(a)$ Average harvested energy per user (top) and $(b)$ average energy outage per user (bottom), as a function of $|\mathcal{S}|$ for $M=8$.}
	\label{Fig9}	 	
\end{figure}
Fig.~\ref{Fig9} shows the average  harvested energy and outage per user as a function of the number of sensors. Since the average performance of the CSI-free schemes does not depend on $|\mathcal{S}|$, it appears as straight lines, while the performance of the CSI-based schemes gets asymptotically worse with $|\mathcal{S}|$. The latter is because as $|\mathcal{S}|$ increases, selecting the best antenna under the $\mathrm{OA-CSI}$ scheme converges to select any random antenna, e.g., $\mathrm{OA}$, while for the $\mathrm{AA-CSI}$ scheme, the chances of creating sharp and strong beams capable of reaching all the users, decreases. Observe that the $\mathrm{SA}$ scheme outperforms the $\mathrm{OA}$ strategy in terms of average harvested energy under sensitivity and saturation phenomena, which is different from the ideal case discussed in Subsection~\ref{linear}. Also, this scheme performs the best in terms of energy outage because of the relatively small values of $\eta\varpi_1$ as discussed in Fig.~\ref{Fig6}. This means that in order to get the advantages of the diversity order and low-variance of the $\mathrm{SA}$ scheme, sensors have to be able of operating at very low power levels and requiring uninterrupted operation. Additionally, the CSI-based schemes outperform their CSI-free counterparts in both metrics and the gap is significant when powering a relatively small number of sensors. However, as $|\mathcal{S}|$ increases (and even for $|\mathcal{S}|=1$) the CSI acquisition procedures and associated energy expenditures become a problem and limit the practical benefits of these schemes. In fact, according to the analyses in Subsection~\ref{linear} the performance gap between the $\mathrm{AA}$ and $\mathrm{AA-CSI}$ schemes is even expected to decrease if any of $M,\kappa,\tau$ increases. 
Moreover, notice that the approximations obtained in Subsection~\ref{non} for $\mathcal{|S|}=1$ hold for any value of $|\mathcal{S}|$ when using the CSI-free schemes, and as shown in Fig.~\ref{Fig9}, the expressions are particularly accurate for the $\mathrm{OA}$ and $\mathrm{AA}$ schemes, while not for the $\mathrm{SA}$ strategy for the reasons given in Theorems~\ref{the2} and \ref{the3}.

\begin{figure}[t!]
	\centering
	\subfigure{\includegraphics[width=0.95\columnwidth]{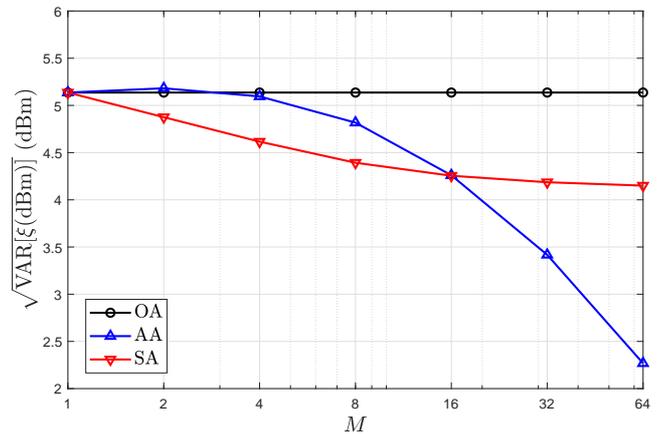}}
	\caption{Standard deviation of the harvested energy in dBm as a function of $M$.}
	\label{Fig10}	 	
\end{figure}
Finally, as a metric of fairness in the wireless powering we show in Fig.~\ref{Fig10} the standard deviation of the harvested energy across time (channel realization) and space (sensor location)  as a function of $M$. Herein we focus only on the CSI-free schemes\footnote{Notice that schemes based on CSI can be designed for optimizing the wireless powering fairness instead of the total harvested energy which is the goal of the CSI-based schemes utilized in this work as benchmarks.}, hence, the results are independent of the number of sensors. Notice that for the considered system setup the $\mathrm{SA}$ scheme is not only the more reliable (as shown in Fig.~\ref{Fig9}b) but also the fairest for relatively small number of antennas $M$, while as $M$ increases the $\mathrm{AA}$ scheme becomes the clear winner. The latter is because more sensors are induced to work near or in the saturation region, hence, outperforming all the other CSI-free schemes in terms of average harvested energy, outage and fairness. Finally, as the $\mathrm{OA}$ scheme utilizes only one antenna, its performance is independent of $M$.
\section{Conclusion and Future Work}\label{conclusions}
In this paper, we analyzed three WET strategies that a dedicated power station equipped with multiple antennas can utilize when powering a large set of single-antenna devices. These strategies operate without any CSI, which is a very practical assumption since it is very challenging to acquire CSI at the transmit side in low-power WET systems with large number of users. The distribution of the harvested energy when using each of the multi-antenna schemes is attained supposing Rician fading channels under sensitivity and saturation EH phenomena and compared to other two CSI-based benchmark strategies. 
We investigated the impact of spatial correlation among the antenna elements evidencing its benefits only when transmitting simultaneously with equal power by all antennas, while under the operation of the other schemes the system performance may reduce considerably when the spatial correlation increases.
When powering a given sensor, we found that: the switching antennas strategy, where only one antenna with full power transmits at a time, guarantees the lowest variance in the harvested energy, thus providing the most predictable energy source, and it is particularly suitable under highly sensitive EH hardware and when operating under poor LOS conditions.
Moreover, while under NLOS it is better switching antennas, when LOS increases it is better to transmit simultaneously with equal power by all antennas. Also, latter strategy provides the greatest fairness in the wireless powering of multiple sensors when there is a large number of antennas since devices are more likely to operate near or in saturation. Numerical and simulation results validate our findings and demonstrate the suitability of the CSI-free over the CSI-based strategies as the number of devices increases. All these are fundamental results that can be used when designing practical WET systems.

\balance
\bibliographystyle{IEEEtran}
\bibliography{IEEEabrv,references}
\end{document}